\documentclass[runningheads]{llncs}
\usepackage{graphicx}
\usepackage{proof} % compact proof trees for examples
\usepackage{hyperref} % links
\usepackage{verbatim} % multiline comments
\usepackage{xcolor} 
\usepackage{mathtools}
\usepackage{amssymb}
\usepackage{amsmath}
\usepackage{stmaryrd}
\usepackage{xspace}
\usepackage{listings}
\usepackage{color}
\usepackage{amssymb}
\usepackage{pifont}
\usepackage{stmaryrd}
\usepackage{listings}
\usepackage{color}

\newif\ifsubmit
\submittrue

\ifsubmit
\newcommand{\EZ}[1]{{#1}} 

\newcommand{\DA}[1]{{#1}} 
\newcommand{\DAComm}[1]{} 
\newcommand{\EZComm}[1]{}
\newcommand{\PBComm}[1]{}
\else
\newcommand{\EZ}[1]{\textcolor{blue}{#1}}  
 
\newcommand{\DA}[1]{\textcolor{orange}{#1}} 
\newcommand{\DAComm}[1]{{\scriptsize\textcolor{orange}{[\bf{Davide: }#1}]}}
\newcommand{\EZComm}[1]{{\scriptsize\textcolor{blue}{[\bf{Elena: }#1}]}}
\newcommand{\PBComm}[1]{{\scriptsize\textcolor{purple}{[\bf{Pietro: }#1}]}}
\fi

\newcommand{\refToFigure}[1]{Fig.~\ref{fig:#1}}
\newcommand{\refToSection}[1]{Sect.~\ref{sect:#1}}
\newcommand{\refToTheorem}[1]{Theorem~\ref{theo:#1}}
\newcommand{\refToLemma}[1]{Lemma~\ref{lemma:#1}}

\newcommand{\refToRule}[1]{{\small \textsc{(#1)}}}
\newcommand{\Space}{\hskip 0.8em}
\newcommand{\BigSpace}{\hskip 1.5em}

\newcommand{\dom}[1]{\mathit{dom}(#1)}
\newcommand{\Tuple}[1]    {({#1})}
\newcommand{\Pair}[2]     {\Tuple{{#1},{#2}}}
\newcommand{\fv}[1]{\textit{fv}(#1)}
\newcommand{\vars}[1]{\textit{vars}(#1)}
 
\newcommand{\NamedRule}[4]{{\tiny\textsc{({#1})}}\displaystyle\frac{#2}{#3}\ \begin{array}{l} #4 \end{array}}

\newcommand{\infeRule}[3]{\infer[\tiny\textsc{({#1})}]{#3}{#2}}

\newenvironment{grammatica}{$\begin{array}[t]{llll}}{\end{array}$}
\newcommand{\produzione}[3]{#1&{:}{:}=&#2 & \mbox{{\small{#3}}}}
\newcommand{\terminale}[1]{\texttt{#1}}
\newcommand{\nonterminale}[1]{\textit{#1}}

\newcommand{\fd}{\nonterminale{fd}}
\newcommand{\f}{\nonterminale{f}}
\newcommand{\undef}{\nonterminale{undef}}
\newcommand{\g}{\nonterminale{g}}
\newcommand{\fdBar}{\overline{\fd}}
\newcommand{\x}{\nonterminale{x}} %variabile generica
\newcommand{\y}{\nonterminale{y}} %variabile generica
\newcommand{\z}{\nonterminale{z}} %variabile generica
\newcommand{\xBar}{\overline{\x}}
\newcommand{\E}{\nonterminale{e}}
\newcommand{\EBar}{\overline{\E}}
\newcommand{\s}{\nonterminale{se}}
\newcommand{\n}{\nonterminale{ne}}
\newcommand{\be}{\nonterminale{be}}
\newcommand{\true}{\texttt{true}}
\newcommand{\false}{\texttt{false}}
\newcommand{\cons}[2]{{#1}\mathbin{\terminale{:}}{#2}}
\newcommand{\tail}[1]{{#1}{\char`\^}}
\newcommand{\IfThenElse}[3] {{\terminale{if}\mathrel{#1}\terminale{then}}\mathrel{#2}{\terminale{else}\mathrel{#3}}}
\newcommand{\PW}[3]{#1{#2}#3}
\newcommand{\nop}{\nonterminale{nop}} %% numeric op
\newcommand{\op}{\nonterminale{op\ }}
\newcommand{\pw}[1]{[{#1}]}
\newcommand\IL{\|}
\newcommand{\pwnop}{\pw{\nop}}
\newcommand{\pwop}{\pw{\nonterminale{op}}}

\newcommand{\Caps}[2]{\Pair{#1}{#2}}

\newcommand{\call}{\mathit{c}}
\newcommand{\val}{\nonterminale{v}}
\newcommand{\vBar}{\overline{\val}}
\newcommand{\callEnv}{\tau}
\newcommand{\emptyMap}{\ensuremath{\emptyset}}
\newcommand{\mapEnv}{\rho}
\newcommand{\mapEnvPrime}{{\mapEnv'}}
\newcommand{\mapEnvU}{{\widehat{\mapEnv}}}
\newcommand{\sv}{\nonterminale{s}}
\newcommand{\nv}{\nonterminale{n}}
\newcommand{\bv}{\nonterminale{b}}
\newcommand{\opsem}[5]{{#1,#2,#3}\!\Downarrow\!\Pair{#4}{#5}}

\newcommand{\bigmapUnion}[2]{\bigsqcup_{#1} #2}
\newcommand{\Subst}[3]   {#1 [#2/#3]}
\newcommand{\Update}[3]   {#1 \{#2\mapsto#3\}}
\newcommand{\At}[4]{\atfun_#1(#2,#3)=#4}
\newcommand{\Atop}[3]{\atfun_#1(#2,#3)}
\newcommand{\atfun}{\mathit{at}}
\newcommand{\wdOpSem}{\mathit{wd}}
\newcommand{\indx}{i}
\newcommand{\jndx}{j}

\newcommand{\stream}{\sigma} %stream semantica
\newcommand{\eval}[2]{\sem{#1}#2}
\newcommand{\subst}{\theta}

\newcommand{\AppMap}[2]{ #1[#2] }
\newcommand{\sem}[1]{{\llbracket#1\rrbracket}}

\newcommand{\WD}[2]{\mathsf{wd}_\mapEnv(#1,#2)}
\newcommand{\WDUpdated}[2]{\mathsf{wd}_{\mapEnv'}(#1,#2)}
\newcommand{\map}{\mathit{m}}
\newcommand{\incrMap}{\map^{+1}}
\newcommand{\decrMap}{\map^{-1}}
\newcommand{\wdRel}[3]{#1\bowtie(#2,#3)}

\newcommand{\AT}[2]{\atfun_\mapEnv(#1,#2)}
\newcommand{\premise}{\vdash}
\newcommand{\premisestar}{\vdash^\star}

\newcommand{\len}[1]{\mathit{length(#1)}}

\newcommand{\owd}[3]{\mathsf{owd}_\mapEnv(#1,#2,#3)}
\newcommand{\owdUpdated}[3]{\mathsf{owd}_{\mapEnv'}(#1,#2,#3)}
\newcommand{\p}{\pi}
\newcommand{\optmap}{\mathsf{m}}
\newcommand{\emptyPath}{\ensuremath{\epsilon}}
\newcommand{\appOp}{\ensuremath{\mathbin{\cdot}}} %% append operator
\newcommand{\sumFrom}{\ensuremath{\mathit{sum}}} %% sum all indexes in the path

\begin{document}
\title{\EZ{Enhancing expressivity\\ of checked corecursive streams (extended version)}
%\thanks{Supported by organization x.}
}
%
%\titlerunning{Enhancing expressivity of checked corecursive streams}
%

\author{Davide Ancona \and Pietro Barbieri \and Elena Zucca}
\authorrunning{D. Ancona et al.}
% First names are abbreviated in the running head.
% If there are more than two authors, 'et al.' is used.
%
\institute{DIBRIS, University of Genova
} 
\maketitle             
\begin{abstract}
We propose a novel approach to stream definition and manipulation.\EZComm{cut:, which reconciles, in a sense, lazy evaluation and regular corecursion} Our solution is based on two key ideas.
\EZ{Regular corecursion, which avoids non termination by detecting cyclic calls,} is enhanced, by allowing in equations defining streams other operators besides the stream constructor. In this way, some non-regular streams are definable. Furthermore, execution includes a runtime check to ensure that the stream generated by a function call is well-defined, in the sense that access to an arbitrary index always succeeds. \EZ{We extend the technique beyond the simple stream operators considered in previous work, notably by adding an \emph{interleaving} combinator which has a non-trivial recursion scheme. }
\end{abstract}

\PBComm{17 pagine esclusi reference e appendice,  submission entro 20/02 AoE}

\section{Introduction}\label{sect:intro}

Applications often deal with data structures which are conceptually infinite; among those \emph{data streams} ({unbounded} \EZComm{mi sembra misleading, al limite toglierei la frase in parentesi} sequences of data) are a paradigmatic example, important in several application domains as the Internet of Things.  
Lazy evaluation is a well-established and widely-used solution to data stream generation and processing, supported, e.g., in Haskell, and in most stream libraries offered by mainstream languages, as \lstinline!java.util.stream!. In this approach, data streams can be defined as the result of an arbitrary function. For instance, in Haskell we can write
\begin{lstlisting}
one_two = 1:2:one_two$\BigSpace\texttt{-- 1:2:1:2:1: ...}$
from n = n:from(n+1)$\BigSpace\texttt{-- n:n+1:n+2: ...}$
\end{lstlisting}
Functions which only need to inspect a finite portion of the structure, e.g., getting the $i$-th element, can be correctly implemented, thanks to the lazy evaluation strategy as exemplified below. % On the other hand, functions which need to inspect the whole structure, e.g., checking that all elements are positive, always diverge.
%This behaviour is 
\begin{lstlisting}
get_elem 3 (one_two)$\BigSpace\texttt{-- evaluates to 2}$
get_elem 3 (from 5)$\BigSpace\texttt{-- evaluates to 7}$
\end{lstlisting}
%% Note that the second call diverges even though, in practice, only \emph{two} checks would be enough, since the stream \EZ{is periodic, a.k.a. \emph{regular} following the terminology in \cite{Courcelle83}, meaning that the term \texttt{1:2:1:2:1: ...} is infinite but has a finite number of subterms.} 

More recently, another approach has been proposed \cite{Jeannin17,SimonBMG07,DagninoAZ20,AnconaBDZ20}, called \emph{regular corecursion}, which exploits  the fact that streams as \lstinline{one_two} above  are periodic, a.k.a. \emph{regular} following the terminology in \cite{Courcelle83}, meaning that the term \texttt{1:2:1:2:1: ...} is infinite but has a finite number of subterms.
Regular streams can be actually represented at runtime by a finite set of equations involving only the stream constructor, in the example $\x = 1 : 2 : \x$.
{Furthermore, function definitions are \emph{corecursive}, meaning that they do not have the standard inductive semantics; indeed, even though the evaluation strategy is call-by-value, thanks to the fact that pending function calls are tracked, cyclic calls are detected, avoiding in this case non-termination. 

For instance, with regular corecursion\footnote{Here we use the syntax of our calculus, where, differently from Haskell, functions are uncurried, that is, take as arguments possibly empty tuples delimited by parentheses.} we have:
\begin{lstlisting}
one_two() = 1:2:one_two()
from(n) = n:from(n+1)
get_elem(3,one_two())$\BigSpace\texttt{-- evaluates to 2}$
get_elem (3,from(5))$\BigSpace\texttt{-- leads to non-termination}$
\end{lstlisting}
Despite their differences, in both approaches programmers are allowed to write intuitively ill-formed definitions such as {\lstinline!bad_stream() = bad_stream()!; any access to indexes of the stream returned by this function  leads to non-termination both with lazy evaluation and regular corecursion.
However, while in the regular case it is simple to reject the result of calling \lstinline{bad_stream} by checking a guardedness syntactic condition, the Haskell compiler does not complain if one calls such a function.
%%In other words, the cyclic call is \emph{not} detected since it occurs behind a \lstinline{tail} operator. 
%% The first two columns of the table below summarize the considerations so far:
%% \begin{quote}
%% \begin{small}
%% \begin{tabular}{l|c|c|c|}
%%      & Lazy Evaluation & Regular Corecursion & Our approach \\
%% \hline
%%     \texttt{getElem (3,from(5))} & \cmark & \xmark & \cmark  \\
%%     \texttt{allPos (one\_two())} & \xmark & \cmark & \cmark \\
%%     \texttt{bad\_stream()} & \xmark & \xmark & \cmark \\
%% \hline
%% \end{tabular}
%% \end{small}
%% \end{quote}
In this paper, we propose a novel approach to stream {generation} and manipulation, \EZ{providing, in a sense, a middle way between} those described above. Our solution is based on two \mbox{key ideas:}
\begin{itemize}
\item Corecursion is enhanced, by allowing in {stream equations other typical operators besides the stream constructor}; in this way, some non-regular streams are supported. For instance, we can define \lstinline{from(n)=n:(from(n)[+]repeat(1))}, with \lstinline{[+]} the pointwise addition and \lstinline{repeat} defined by
\lstinline{repeat(n)=n:repeat(n)}.
\item Execution includes a runtime check \EZ{which rejects the stream generated by a function call if it is ill-formed, in the sense that access to an index could possibly diverge}. For instance, the call \lstinline{bad_stream()} raises a runtime error.
\end{itemize}
In this way we {achieve a convenient} trade-off between expressive power and reliability; indeed, we do not have the full expressive power of Haskell, where we can manipulate streams generated as results of arbitrary functions, but, clearly, the well-definedness check described above would be not decidable. 
On the other hand, we {significantly} augment the expressive power of regular corecursion, allowing {several} significant non-regular streams, at the price of making the well-definedness check non-trivial, but still decidable. 

The main {formal} results are (1) \refToTheorem{iff} stating the soundness of the runtime check;   (2) \refToTheorem{wd-eq-gen} stating that the
{optimized definition of the runtime check in \refToSection{opt-wd} is equivalent to the simpler one given in \refToSection{wd}}.
  {In particular, for contribution (1) {the interleaving operator requires a more involved proof in comparison with} \cite{AnconaBZ21} (see \refToSection{conclu}), while for (2) we show that the optimized definition improves the time complexity from $O(N^2)$ to $O(N \log N)$.}
  
In \refToSection{calculus} we formally define the calculus, in \refToSection{examples} we show examples, in \refToSection{wd} we define the well-formedness check, and in \refToSection{opt-wd} its optimized version. 
Finally, in \refToSection{conclu} we discuss related and further work. {The Appendix contains more examples of derivations and omitted proofs.}

\section{Stream calculus}\label{sect:calculus}

\refToFigure{stream-syntax} shows the syntax of the calculus. 
\begin{figure}
\begin{grammatica}
\produzione{\fdBar}{\fd_1\ldots\fd_n}{program}\\
\produzione{\fd}{\f(\xBar) = \s}{function declaration}\\
\produzione{\E}{\s \mid \n \mid \be}{expression}\\
\produzione{\s}{\x \mid \IfThenElse{\be}{\s_1}{\s_2} \mid \cons{\n}{\s} \mid \tail{\s} \mid {\PW{\s_1}{\op}{\s_2}} \mid \f(\EBar)}{stream expression}\\
\produzione{\n}{\x \mid \s(\n) \mid {\n_1\mathop{\nop}\n_2} \mid 0 \mid 1 \mid 2 \mid ...}{numeric expression}\\
\produzione{\be}{\x \mid \true \mid \false \mid ...}{boolean expression}\\
\produzione{\op}{\pwnop \mid \IL}{binary stream operator}\\
\produzione{\nop}{+\ \mid\ -\ \mid\ *\ \mid\ /}{numeric operator}
\end{grammatica}
\caption{Stream calculus: syntax}
\label{fig:stream-syntax}
\end{figure}

A program is a sequence of (mutually recursive) function declarations, for simplicity assumed to only return streams. Stream expressions are variables, conditionals, expressions built by stream operators, and function calls. We consider the following stream operators: constructor (prepending a numeric element), tail, pointwise {arithmetic} operators, {and interleaving}. Numeric expressions  include the access to the $i$-th\footnote{For simplicity, here indexing and numeric expressions coincide.\EZComm{cut:, even though indexes are expected to be natural numbers, while values in streams can range over a larger numeric domain.}}  element of a stream. We use $\fdBar$ to denote a sequence $\fd_1, \dots, \fd_n$ of function declarations, and analogously for other sequences.

The operational semantics, given in \refToFigure{stream-sem}, is based on two key ideas:
\begin{enumerate}
\item some infinite streams can be represented in a finite way
\item evaluation keeps trace of already considered function calls 
\end{enumerate}

\begin{figure}
\begin{small}
\begin{grammatica}
\produzione{\call}{\f(\vBar)}{{(evaluated)} call}\\
\produzione{\val}{\sv \mid \nv \mid \bv}{value}\\
\produzione{\sv}{\x \mid \cons{\nv}{\sv} \mid \tail{\sv} \mid \PW{\sv_1}{\op}{\sv_2} }{{{(open)} stream value}}\\
\produzione{{\indx, \nv}}{0 \mid 1 \mid 2 \mid ...}{{index, numeric value}}\\
\produzione{\bv}{\true \mid \false}{boolean value}\\
\produzione{\callEnv}{\call_1\mapsto\x_1\ \ldots\ \call_n\mapsto\x_n \Space (n\geq0)}{call trace}\\
\produzione{\mapEnv}{{\x_1\mapsto\sv_1 \ldots \x_n\mapsto\sv_n} \Space (n\geq0)}{environment}\\
\end{grammatica}
\\[2ex]

\hrule 

$\begin{array}{l}
  \\
\NamedRule{val}{}{\opsem{\val}{\mapEnv}{\callEnv}{\val}{\mapEnv}}{}
\Space
\NamedRule{if-t}{
\opsem{\be}{\mapEnv}{\callEnv}{\true}{\mapEnv} \Space
\opsem{\s_1}{{\mapEnv}}{\callEnv}{\sv}{\mapEnv'}
}{\opsem{\IfThenElse{{\be}}{\s_1}{\s_2}}{\mapEnv}{\callEnv}{\sv}{{\mapEnv'}}
}{}
\Space
\NamedRule{if-f}{
\opsem{\be}{\mapEnv}{\callEnv}{\false}{{\mapEnv}} \Space
\opsem{\s_2}{{\mapEnv}}{\callEnv}{\sv}{{\mapEnv'}}
}{\opsem{\IfThenElse{{\be}}{\s_1}{\s_2}}{\mapEnv}{\callEnv}{\sv}{{\mapEnv'}}
}{}
\\[6ex]
\NamedRule{cons}{
\opsem{\n}{\mapEnv}{\callEnv}{\nv}{\mapEnv}\Space
\opsem{\s}{\mapEnv}{\callEnv}{\sv}{\mapEnv'}}{
\opsem{\cons{\n}{\s}}{\mapEnv}{\callEnv}{\cons{\nv}{\sv}}{\mapEnv'}}
{} \BigSpace
\NamedRule{tail}{
\opsem{\s}{\mapEnv}{\callEnv}{\sv}{\mapEnv'}}{
\opsem{\tail{\s}}{\mapEnv}{\callEnv}{\tail{\sv}}{\mapEnv'}}
{}
\BigSpace
{
\NamedRule{op}{
\opsem{\s_1}{\mapEnv}{\callEnv}{\sv_1}{\mapEnv_1}\Space
\opsem{\s_2}{\mapEnv}{\callEnv}{\sv_2}{\mapEnv_2}}{
\opsem{\PW{\s_1}{\op}{\s_2}}{\mapEnv}{\callEnv}{\PW{\sv_1}{\op}{\sv_2}}{\mapEnv_1\sqcup\mapEnv_2}}
{}}
\\[6ex]
{\NamedRule{args}{  \begin{array}{l}
    \opsem{\E_i}{\mapEnv}{\callEnv}{\val_i}{{\mapEnv_i}}\Space \forall i \in 1..n
    \BigSpace
    \opsem{\f(\vBar)}{{\mapEnvU}}{\callEnv}{\sv}{\mapEnvPrime}
  \end{array}
}{ \opsem{\f(\EBar)}{\mapEnv}{\callEnv}{\sv}{\mapEnvPrime}}
{ \EBar=\E_1,\ldots,\E_n\ \mbox{not of shape}\ \vBar\\
\vBar=\val_1,\ldots,\val_n\\
\mapEnvU = \bigmapUnion{i \in 1..n}{\mapEnv_i}}}

\\[6ex]
{\NamedRule{invk}{  \begin{array}{l}
    \opsem{\Subst{\s}{\vBar}{\xBar}}{\mapEnv}{\Update{\callEnv}{\f(\vBar)}{\x}}{\sv}{\mapEnvPrime}
  \end{array}
}{ \opsem{\f(\vBar)}{\mapEnv}{\callEnv}{{\x}}{\Update{\mapEnvPrime}{\x}{\sv}}}
{ {\f(\vBar)\not\in\dom{{\callEnv}}}\\
 \x\ \mbox{fresh}\\ 
\mathit{fbody}(\f)=\Pair{\xBar}{\s}\\
\wdOpSem(\mapEnvPrime,\x,\sv)}}
\BigSpace
{\NamedRule{corec}{
}{ \opsem{\f(\vBar)}{\mapEnv}{\callEnv}{\x}{\mapEnv}}  
{
{\callEnv}({f(\vBar)})=\x\\
}}
\\[8ex]
\NamedRule{at}{
  \opsem{\s}{\mapEnv}{\callEnv}{\sv}{{\mapEnv'}}\Space  
  \opsem{\n}{\mapEnv}{\callEnv}{\indx}{\mapEnv}\Space
}{ \opsem{\s(\n)}{\mapEnv}{\callEnv}{\nv}{ {\mapEnv} } }
{  \At{\mapEnvPrime}{\sv}{\indx}{\nv}}
\\[6ex]
\hline
\\
\NamedRule{at-var}
{\At{\mapEnv}{\mapEnv(\x)}{\indx}{\nv'}}
{\At{\mapEnv}{\x}{\indx}{\nv' }}
{  }
\BigSpace
\NamedRule{at-cons-0}
{}
{\At{\mapEnv}{\cons{\nv}{\sv}}{0}{\nv }}
{  }
\BigSpace
{\NamedRule{at-cons-succ}
{\At{\mapEnv}{\sv}{\indx}{\nv'}}
{\At{\mapEnv}{\cons{\nv}{\sv}}{\indx+1}{\nv'}}
{}}
\\[6ex]
\NamedRule{at-tail}
{\At{\mapEnv}{\sv}{\indx+1}{\nv}}
{\At{\mapEnv}{\tail{\sv}}{\indx}{\nv}}
{}%%{\nv'\geq 0} useless side-condition, \nv meta-variabes range over natural numbers
\BigSpace
\NamedRule{{at-nop}}
{\At{\mapEnv}{\sv_1}{\indx}{\nv_1}\Space{\At{\mapEnv}{\sv_2}{\indx}{\nv_2}}}
{\At{\mapEnv}{\PW{\sv_1}{\pwnop}{\sv_2}}{\indx}{\nv_1\mathbin{\nop}\nv_2}}
{}
\\[6ex]
{
\NamedRule{at-$\IL$-even}
{\At{\mapEnv}{\sv_1}{\indx}{\nv}}
{\At{\mapEnv}{\sv_1\IL\sv_2}{2\indx}{\nv}}
{}
\BigSpace
\NamedRule{at-$\IL$-odd}
{\At{\mapEnv}{\sv_2}{\indx}{\nv}}
{\At{\mapEnv}{\sv_1\IL\sv_2}{2\indx+1}{\nv}}
{}}
\end{array}$
\end{small}
\caption{Stream calculus: operational semantics}\label{fig:stream-sem}
\end{figure}

To obtain (1), our approach is inspired by \emph{capsules} \cite{JeanninK12}, which are expressions supporting cyclic references. That is, the \emph{result} of a stream expression is a pair $\Pair{\sv}{\mapEnv}$, where $\sv$ is an \emph{(open) stream value}, built on top of stream variables, numeric values, the stream constructor, the tail destructor, the pointwise arithmetic and the interleaving operators, and $\mapEnv$ is an \emph{environment} mapping variables into stream values.
In this way, cyclic streams can be obtained: for instance, $\Pair{\x}{\x\mapsto\cons{\nv}{\x}}$ denotes the stream constantly equal to 
 $\nv$.
 
 We denote by {$\dom{\mapEnv}$ the domain of $\mapEnv$,} by $\vars{\mapEnv}$ the set of variables occurring in $\mapEnv$, by $\fv{\mapEnv}$ the set of its free variables, that is, $\vars{\mapEnv}\setminus\dom{\mapEnv}$, and say that $\mapEnv$ is \emph{closed} if $\fv{\mapEnv}=\emptyset$, \emph{open} otherwise, and analogously for a result $\Pair{\val}{\mapEnv}$.

To obtain {point (2) above}, evaluation has an additional parameter which is a \emph{call trace}, a map from function calls where arguments are values  (dubbed \emph{calls} for short in the following) into variables. 

Altogether, the semantic judgment has shape $\opsem{\E}{\mapEnv}{\callEnv}{\val}{\mapEnvPrime}$, {where 
$\E$ is the expression to be evaluated, $\mapEnv$ the current environment defining possibly cyclic stream values that can occur in $\E$,
  $\callEnv$  the call trace, and $\Pair{\val}{\mapEnvPrime}$ the result.}
  The semantic judgments should be indexed by an underlying (fixed) program, 
 omitted for sake of simplicity. Rules use the following auxiliary definitions:
\begin{itemize}
\item $\mapEnv\sqcup\mapEnv'$ is the union of two environments, which is well-defined if they have disjoint domains; $\Update{\mapEnv}{\x}{\sv}$ is the environment which gives $\sv$ on $\x$, coincides with $\mapEnv$ elsewhere; we use analogous notations for call traces.
\item $\Subst{\s}{\vBar}{\xBar}$ is obtained by {parallel substitution of} variables $\xBar$ {with} values $\vBar$.
\item $\mathit{fbody}(\f)$ returns the pair of the parameters and the body of the declaration of $\f$, if any, {in the assumed program}.
\end{itemize}

Intuitively, a closed result $\Pair{\sv}{\mapEnv}$ is well-defined if it denotes a unique stream, and {a closed environment} $\mapEnv$ is well-defined if, for each $\x\in\dom{\mapEnv}$, $\Pair{\x}{\mapEnv}$ is well-defined. In other words, the corresponding set of equations admits a unique solution. For instance, the {environment
  ${\{\x\mapsto\x\}}$} is not well-defined, since it is undetermined (any stream satisfies the equation {$\x=\x$}); the {environment $\{\x\mapsto\x[+]\y,\y\mapsto\cons{1}{\y}\}$} is not well-defined as well, since it is undefined ({the two equations $\x=\x\mapsto\x[+]\y,\y=\cons{1}{\y}$ admit no solutions for $x$}). 
This notion can be generalized to open results and environments, assuming that free variables denote unique streams, as will be formalized in \refToSection{wd}. 

Rules for values and conditional are straightforward. In rules \refToRule{{cons}}, \refToRule{tail} and \refToRule{op}, arguments are evaluated and the stream operator is applied without any further evaluation. That is, we treat all these operators as constructors.

The rules for function call are based on a mechanism of cycle detection \cite{AnconaBDZ20}. Evaluation of arguments is handled by a separate rule \refToRule{args}, whereas the following two rules handle (evaluated) calls.

Rule \refToRule{invk} is applied when a call is considered for the first time, as expressed by the first side condition. The body is retrieved by using the auxiliary function \textit{fbody}, and evaluated in  a call trace where the call has been mapped into a fresh variable. Then, it is checked that adding  the association of such variable with the result of the evaluation of the body keeps the environment well-defined, {as expressed by the judgment $\wdOpSem(\mapEnv,\x,\sv)$, which will be defined in \refToSection{wd}.}
If the check succeeds, then the final result consists of the variable associated with the call and the updated environment. For simplicity, here execution is stuck if the check fails; an implementation should raise a runtime error instead. An example of stuck derivation is presented in the Appendix (\refToFigure{stuck_derivation}).

Rule \refToRule{corec} is applied when a  call is considered for the second time, as expressed by the first side condition.
The variable $\x$ is returned as result. {However, there is no associated value in the environment yet; in other words, the {result $\Pair{\x}{\mapEnv}$ is open} at this point. } This means that $\x$ is undefined until the environment is updated with the corresponding value in rule \refToRule{invk}. However, $\x$ can be safely used as long as the evaluation does not require $\x$ to be inspected; for instance, $\x$ can be safely passed as an argument to a function call.

For instance, if we consider the program \lstinline!f()=g()  g()=1:f()!, then the judgment
  $\opsem{\mbox{\lstinline{f()}}}{\emptyMap}{\emptyMap}{\x}{\mapEnv}$, with $\mapEnv=\{\x\mapsto\y,\y\mapsto \cons{1}{\x}\}$,  is derivable;
  however, while the final result $\Pair{\x}{\mapEnv}$ is closed, the derivation contains also judgments with open results, as happens for
$\opsem{\mbox{\lstinline{f()}}}{\emptyMap}{\{\mbox{\lstinline{f()}}\mapsto\x,\mbox{\lstinline{g()}}\mapsto\y\}}{\x}{\emptyMap}$ and $\opsem{\mbox{\lstinline{g()}}}{\emptyMap}{\{\mbox{\lstinline{f()}}\mapsto\x\}}{\y}{\{\y\mapsto\cons{1}{\x}\}}$. {The full derivation is presented in the Appendix (\refToFigure{derivation1}), together with another example (\refToFigure{derivation2}).}

%As another example, if we consider the program \lstinline!f()=g(2:f())  g(s)=1:s!, then the derivation of the judgment
%$\opsem{\mbox{\lstinline{f()}}}{\emptyMap}{\emptyMap}{\x}{\mapEnv}$ with $\mapEnv=\{\x\mapsto\y,\y\mapsto \cons{1}{\cons{2}{\x}}\}$
%  is built on top of the derivation of $\opsem{\mbox{\lstinline{g(2:}\x\lstinline{)}}}{\emptyMap}{\{\mbox{\lstinline{f()}}\mapsto\x\}}{\y}{\{\y\mapsto\cons{1}{\cons{2}{\x}}\}}$, corresponding to the evaluation of \lstinline{g(2:}\x\lstinline{)} where $\x$ is an operand of the stream constructor
%whose result is passed as argument to the call to \lstinline{g}, despite $\x$ is not {defined yet}. The full derivation is presented in the Appendix (\refToFigure{derivation2}).

  Finally, rule \refToRule{at} computes the $\indx$-th element of a stream expression. After evaluating the arguments, the result is obtained by the auxiliary judgment $\At{\mapEnv}{\sv}{\indx}{\nv}${, whose straightforward definition is at the bottom of the figure.
    {Rules \refToRule{at-$\IL$-even} and \refToRule{at-$\IL$-odd} define the behaviour of the interleaving operator, which merges two streams together by alternating their elements.}

    \DA{When evaluating $\atfun_\mapEnv(\sv,\indx)$, if $\sv$} is a variable free in the environment, then execution is stuck; again, \mbox{an implementation should raise a runtime error instead. }

%If the stream value is built by the constructor, then the result is the first element of the stream if the index is $0$, rule \refToRule{at-cons-0}; otherwise, the {evaluation} is recursively {propagated to} its tail
%{with} the predecessor {index}, rule \refToRule{{at-cons-succ}}. Conversely, if the stream is built by the tail operator, rule \refToRule{at-tail}, then
%the {evaluation} is recursively {propagated to} the stream argument {with} the successor {index}. If the stream is built by a pointwise operation, rule \refToRule{at-nop}, then {the evaluation is recursively propagated to the operands with the same index and then the corresponding arithmetic operation
%is computed on the results}. Finally, if the stream is built by the interleaving operation, we have two mutually exclusive rules. Rule \refToRule{at-$\IL$-even} is used for even indexes $2i$ and propagates the evaluation to the left-hand side stream at index $i$; analogously, for odd indexes $2i+1$, rule \refToRule{at-$\IL$-odd} is applied and the evaluation is propagated to the right-end side stream at index $i$.

\section{Examples}\label{sect:examples}
First we show some simple examples, to explain how {corecursive definitions} work. Then we provide some more significant examples.

Consider the following function declarations:
\begin{lstlisting}
repeat(n) = n:repeat(n) 
one_two() = 1:two_one()
two_one() = 2:one_two()
\end{lstlisting}
With the standard semantics of recursion, the calls, e.g., \lstinline!repeat(0)! and \lstinline!one_two()! lead to non-termination. Thanks to {corecursion}, instead, these calls terminate, producing as result $\mathtt{\Caps{\x}{\{\x\mapsto\cons{0}{\x}\}}}$, and $\mathtt{\Caps{\x}{{\{}\x\mapsto\cons{1}{\y},\y\mapsto\cons{2}{\x}{\}}}}$, respectively. Indeed, when initially invoked, the call  \lstinline!repeat(0)! is added in the call trace with an associated fresh variable, say $\x$. In this way, when evaluating the body of the function, the recursive call is detected as cyclic, the variable $\x$ is returned as its result, and, finally, the stream value $\mathtt{\cons{0}{\x}}$ is associated in the environment {with} the result $\x$ of the initial call. In the sequel, we will use \lstinline![k]! as a shorthand for \lstinline!repeat($k$)!. The evaluation of \lstinline!one_two()! is {analogous, except} that another fresh variable $\y$ is generated for the intermediate call \lstinline!two_one()!. The formal derivations are given below.
\begin{small}
$$\begin{array}{l}
\infeRule{invk}{
  \infeRule{cons}{
   {\infeRule{value}{}{}} &
    %{\infeRule{value}{}{\opsem{\mathtt{0}}{\emptyset}{\{\mathtt{repeat(0)}\mapsto\x\}}{0}{\emptyset}}\quad}
    \infeRule{corec}{
}{\opsem{\mathtt{repeat(0)}}{\emptyset}{\{\mathtt{repeat(0)}\mapsto\x\}}{\x}{\emptyset}}
}{\opsem{\mathtt{\cons{0}{repeat(0)}}}{\emptyset}{\{\mathtt{repeat(0)}\mapsto\x\}}{\cons{0}{\x}}{\emptyset}}}
{\opsem{\mathtt{repeat(0)}}{\emptyset}{\emptyset}{\x}{\{\x\mapsto\cons{0}{\x}\}}}
\end{array}$$

$$\begin{array}{l}
\infeRule{invk}{
  \infeRule{cons}{
  {\infeRule{value}{}{}} &
    %{\NamedRuleSimple{value}{}{\opsem{\mathtt{1}}{\emptyset}{\{\mathtt{one\_two()}\mapsto\x\}}{1}{\emptyset}}\quad}
\infeRule{invk}{
  \infeRule{cons}{
  {\infeRule{value}{}{}} &
    %{\NamedRuleSimple{value}{}{\opsem{\mathtt{2}}{\emptyset}{\{\mathtt{one\_two()}\mapsto\x,\ \mathtt{two\_one()}\mapsto\y\}}{2}{\emptyset}}\quad}
\infeRule{corec}{
}{\opsem{\mathtt{one\_two()}}{\emptyset}{\{\mathtt{one\_two()}\mapsto\x,\ \mathtt{two\_one()}\mapsto\y\}}{\x}{\emptyset}}
}{\opsem{\mathtt{2:one\_two()}}{\emptyset}{\{\mathtt{one\_two()}\mapsto\x,\ \mathtt{two\_one()}\mapsto\y\}}{2:\x}{\emptyset}}
}{\opsem{\mathtt{two\_one()}}{\emptyset}{\{\mathtt{one\_two()}\mapsto\x\}}{\y}{\{\y\mapsto2:\x\}}}
}{\opsem{\mathtt{1:two\_one()}}{\emptyset}{\{\mathtt{one\_two()}\mapsto\x\}}{1:\y}{{\{\y\mapsto2:\x\}}}}}
{\opsem{\mathtt{one\_two()}}{\emptyset}{\emptyset}{\x}{\{\x\mapsto 1:\y,\ \y\mapsto2:\x\}}}
\end{array}$$
\end{small}

For space reasons, we did not report the application of rule \refToRule{value}. In both derivations, note that rule \refToRule{corec} is applied, without evaluating the body of \lstinline{one_two} once more, when the cyclic call is detected. 

The following examples show function definitions {whose calls return} non-regular streams, notably, the natural numbers, {the natural numbers raised to the power of a number, the factorials, the powers of a  number,} the Fibonacci numbers, and the stream obtained by pointwise increment by one.
\begin{lstlisting}  
nat() = 0:(nat()[+][1])
nat_to_pow(n) =           //nat_to_pow(n)(i)=i^n 
  if n <= 0 then [1] else nat_to_pow(n-1)[*]nat()
fact() = 1:((nat()[+][1])[*]fact())
pow(n) = 1:([n][*]pow(n)) //pow(n)(i)=n^i
fib() = 0:1:(fib()[+]fib()^)
incr(s) = s[+][1]
\end{lstlisting}

The definition of \lstinline!nat! uses corecursion, since the recursive call \lstinline!nat()! is cyclic. Hence the call \lstinline!nat()! returns $\Caps{\x}{{\{}\x\mapsto\cons{0}{(\x[+]\y)}, \y\mapsto\cons{1}{\y}{\}}}$.
{The definition of \lstinline!nat_to_pow! is} a standard inductive one {where the argument strictly decreases in the recursive call}. Hence, the call, e.g., \lstinline!nat_to_pow(2)!, returns\\
\centerline{
$\Caps{\x_2}{{\{\x_2\mapsto\x_1[*]\x,\x_1\mapsto\x_0[*]\x, \x_0\mapsto\y, \y\mapsto\cons{1}{\y}, \x\mapsto\cons{0}{(\x[+]\y')}, \y'\mapsto\cons{1}{\y'}\}}}.$}
  The definitions of \lstinline!fact!, \lstinline!pow!, and \lstinline!fib! are corecursive. For instance, the call \lstinline!fact()! returns
  ${\Caps{\z}{\z\mapsto\DA{\cons{1}{((\x[+]\y)[*]z)}}, \x\mapsto\cons{0}{(\x[+]\y'}), \y\mapsto 1:\y, \y'\mapsto 1:\y'}}$. 
 The definition of \lstinline!incr! is non-recursive, hence always converges{, and the call \lstinline{incr(}$\sv$\lstinline{)} returns $\Caps{\x}{{\{}\x\mapsto \sv[+]\y, \y\mapsto\cons{1}{\y}{\}}}$}.
%% The following alternative definition
%% \begin{lstlisting} 
%% incr_reg(s) = (s(0)+1):incr_reg(s^) 
%% \end{lstlisting}
%% relies, instead, on corecursion. Note the difference: the latter version ensures termination {only} for
%%   regular streams, {as in  \lstinline!incr_reg(one_two())! since, eventually, in the recursive call, the expression \lstinline!s^! turns out to denote the initial stream}; however, the computation does not
%%   terminate for non-regular streams, as in \lstinline!incr_reg(nat())!, which, however, converges with \lstinline!incr!.
  
The next few examples show applications of the interleaving operator.
\begin{lstlisting}
dup_occ() = 0:1:(dup_occ() || dup_occ())
\end{lstlisting}
Function \lstinline!dup_occ()! generates the stream which alternates sequences of occurrences of \lstinline!0! and \lstinline!1!, with the number of repetitions of the same number duplicated at each step, that is,  \lstinline!(0:1:0:0:1:1:0:0:0:0...)!.

A more involved example shows a different way to generate the stream of all powers of \lstinline!2! starting from $2^1$:
\begin{lstlisting}
pow_two=2:4:8:((pow_two^^[*]pow_two)||(pow_two^^[*]pow_two^))
\end{lstlisting}
{The following definition is an instance of a schema generating the infinite sequence of labels obtained by a breadth-first visit of an infinite {complete} binary tree where the labels of children are defined in terms of  that of their parent.}
\begin{lstlisting}
bfs_index() = 1:((bfs_index()[*][2])||(bfs_index()[*][2][+][1]))
\end{lstlisting}
\EZComm{mettere esplicitamente lo schema?}\DAComm{si potrebbe se avessimo spazio sufficiente}
{In particular, the root is labelled by \lstinline{1}, and the left and right child of a node with label \lstinline{i} are labelled by \lstinline{2*i} and \lstinline{2*i+1}, respectively. Hence, the generated stream is the sequence of natural numbers starting from \lstinline{1}, as it happens in the array implementation of a binary heap.}

In the other instance below, the root is labelled by \lstinline{0}, and children are labelled with \lstinline{i+1} if their parent has label \lstinline{i} . That is, nodes are labelled \mbox{by their level.}
\begin{lstlisting}
bfs_level() = 0:((bfs_level()[+][1])||(bfs_level()[+][1]))
\end{lstlisting}
In this case, the generated stream is more interesting; indeed, \lstinline!bfs_level()(n) =! \lstinline!floor(log$_2$(n+1))!. 
\leavevmode
  
The following {function computes the stream of partial sums of the first $i+1$ elements of a stream $s$, that is,
  \lstinline!sum($s$)($\indx$)$=\sum_{k=0}^{\indx}s(k)$!}:
\begin{lstlisting}
sum(s) = s(0):(s^[+]sum(s))
\end{lstlisting}
Such a function is useful for computing streams whose elements approximate a series with increasing precision;
for instance, the following function returns the stream of partial sums of the first $i+1$ elements of the Taylor series of the exponential function:
\begin{lstlisting}
sum_expn(n) = sum(pow(n)[/]fact())
\end{lstlisting}
Function \lstinline!sum_expn! calls \lstinline!sum! with the argument \lstinline!pow(n)[/]fact()! corresponding to the stream of terms of the
Taylor series of the exponential;  hence, by accessing the $\indx$-th element of the stream, we have the following approximation of the series:
\begin{lstlisting}
sum_expn($\nv$)($\indx$)$=\displaystyle\displaystyle\sum_{k=0}^{\indx} \frac{{\nv}^k}{k!} = 1+\nv+\frac{{\nv}^2}{2!}+\frac{{\nv}^3}{3!}+\frac{{\nv}^4}{4!}+\cdots+\frac{\nv^{\indx}}{\indx!}$
\end{lstlisting}
Lastly, {we present a couple of examples showing how it is possible to define primitive operations provided
  by IoT platforms for real time analysis of data streams;} we start with \lstinline!aggr(n,s)!, which allows
  aggregation by addition of data in windows of length \lstinline!n!:
\begin{lstlisting}
aggr(n,s) = if n<=0 then [0] else s[+]aggr(n-1,s^) 
\end{lstlisting}
For instance, \lstinline!aggr(3,$\sv$)! returns the stream $\sv'$ s.t. $\sv'(\indx)=\sv(\indx)+\sv(\indx+1)+\sv(\indx+2)$.
%%  If we allow streams of rational numbers, then,
{On} top of \lstinline!aggr!,  we can easily define \lstinline!avg(n,s)!
to compute the stream of average values of \lstinline!s! in windows of length \lstinline!n!:
\begin{lstlisting}
avg(n,s) = aggr(n,s)[/][n]  
\end{lstlisting}

\section{Well-definedness check}\label{sect:wd}

A key feature of our approach is the runtime check ensuring that the stream generated by a function call is well-defined, see the side condition $\wdOpSem(\mapEnvPrime,\x,\sv)$ in \refToRule{invk}; in this section we formally define the corresponding judgment and prove its soundness. Before doing this, we provide, for reference, a formal abstract definition of well-definedness.  

Intuitively, an environment is well-defined if each variable in its domain denotes a unique stream. 
Semantically, a stream $\stream$ is an infinite sequence of numeric values{, that is, a function which returns, for each index $i\geq 0$, the $i$-th element $\stream(i)$.}
Given a result $\Caps{\sv}{\mapEnv}$, we get a stream by instantiating variables in $\sv$ with streams, in a way consistent with $\mapEnv$, and evaluating {operators}.
To make this formal, we need some preliminary definitions.

A \emph{substitution}  {$\subst$} is a function from {a {finite} set of} variables to streams. We denote by $\eval{\sv}{\subst}$ the  stream obtained by applying $\subst$ to $\sv$, and evaluating {operators}, as formally defined below.
\begin{quote}
$\eval{\x}{\subst} =\subst(\x)$\\[1ex]
$(\eval{\cons{\nv}{\sv}}{\subst})(i) = 
\begin{cases}
\nv & i=0 \\ (\eval{\sv}{\subst})(i-1) & i\geq 1 
\end{cases}$\\[1ex]
$(\eval{\tail{\sv}}{\subst})(i) = \eval{\sv}{\subst}(i+1)\BigSpace i\geq 0$\\[1ex]
$(\eval{\PW{\sv_1}{\pwnop}{\sv_2}}{\subst})(i) = \eval{\sv_1}{\subst}(i) \mathbin{\nop} \eval{\sv_2}{\subst}(i)\BigSpace i\geq 0$\\[1ex]
$(\eval{\PW{\sv_1}{\IL}{\sv_2}}{\subst})(2i) = \eval{\sv_1}{\subst}(i)\BigSpace i\geq 0$\\[1ex]
$(\eval{\PW{\sv_1}{\IL}{\sv_2}}{\subst})(2i+1) = \eval{\sv_2}{\subst}(i)\BigSpace i\geq 0$
\end{quote}

Given an environment $\mapEnv$ and a substitution $\subst$ {with domain $\vars{\mapEnv}$}, the substitution  $\AppMap{\mapEnv}{\subst}$ is defined by: 
\begin{quote}
$\AppMap{\mapEnv}{\subst}(\x) = \begin{cases}
\eval{\mapEnv(\x)}{\subst} & \x \in \dom{\mapEnv} \\
\subst(x) & {\x\in \fv{\mapEnv}}
\end{cases}$
\end{quote}
Then, a \emph{solution} of $\mapEnv$ is a substitution $\subst$ such that $\AppMap{\mapEnv}{\subst} = \subst$. 

A closed environment $\mapEnv$ is \emph{well-defined} if it has exactly one solution. 
For instance, ${\{\x\mapsto\cons{1}{\x}\}}$ and ${\{\y\mapsto\cons{0}{(\y [+] \x)},\ \x\mapsto1:\x\}}$ are well-defined{, since their unique solutions map $\x$ to the infinite stream of ones, and $\y$ to the stream of natural numbers, respectively.} Instead, for ${\{\x\mapsto1[+]\x\}}$ there are no solutions. Lastly, an environment can be undetermined{:  for instance, a substitution mapping $\x$ into an arbitrary stream} is a solution of ${\{\x\mapsto\x\}}$.

{An open environment $\mapEnv$ is well-defined if, for each $\subst$ with domain $\fv{\mapEnv}$, it has exactly one solution $\subst'$ such that $\subst\subseteq\subst'$. For instance, the open environment  $\{\y\mapsto\cons{0}{(\y [+] \x)}\}$ is well-defined.} 

%Given a {closed} result $\Caps{\sv}{\mapEnv}$, with $\mapEnv$ well-defined, we define its semantics by  
%$\sem{\sv}{\mapEnv} = \eval{\sv}{\subst}$ for  $\subst =\sol(\mapEnv)$. Then, two stream values $\sv$ and $\sv'$ are \emph{semantically equivalent in $\mapEnv$}   if $\sem{\sv}{\mapEnv}=\sem{\sv'}{\mapEnv'}$.

In \refToFigure{op-wd} we provide the operational characterization of well-definedness. 
\begin{figure}
\begin{small}
\begin{grammatica}
\produzione{\map}{{\x_1\mapsto\nv_1 \ldots \x_n\mapsto\nv_k} \Space (n\geq0)}{map from variables to integer numbers}
\end{grammatica}
\\
\hrule 
$\begin{array}{l}
  \\
\NamedRule{main}
{{\WDUpdated{\x}{\emptyMap}}}
{\wdOpSem(\mapEnv,\x,\val)}
{{\mapEnv'=\Update{\mapEnv}{\x}{\val}}}
\BigSpace    
\NamedRule{wd-var}
{\WD{\mapEnv(\x)}{\Update{\map}{\x}{0}}}
{\WD{\x}{\map} }
{\x\not\in\dom\map}
\\[5ex]
\NamedRule{wd-corec}
{}
{\WD{\x}{\map} }
{\EZ{\x\in\dom{\mapEnv}}\\
\map(x)>0}
\BigSpace
{\NamedRule{wd-delay}
{\WD{\mapEnv(\x)}{\Update{\map}{\x}{0}}}
{\WD{\x}{\map} }
{
\map(x)>0}}
\\[5ex]
{\NamedRule{wd-fv}
{}
{\WD{\x}{\map} }
{\x\not\in\dom\mapEnv}}
\BigSpace
\NamedRule{wd-cons}
{\WD{\sv}{{\incrMap}}}
{\WD{\cons{\nv}{\sv}}{\map}}
{}
\BigSpace
\NamedRule{wd-tail}
{\WD{\sv}{{\decrMap}}}
{\WD{\tail{\sv}}{\map}}
{}
\\[5ex]
\NamedRule{wd-nop}
{\WD{\sv_1}{\map}\Space\WD{\sv_2}{\map}}
{\WD{\PW{\sv_1}{\pwnop}{\sv_2}}{\map}}
{}
\BigSpace
\NamedRule{wd-$\IL$}
{\WD{\sv_1}{\map}\Space\WD{\sv_2}{\incrMap}}
{\WD{\PW{\sv_1}{\IL}{\sv_2}}{\map}}
{}
\end{array}$
\end{small}
\caption{Operational definition of well-definedness}\label{fig:op-wd}
\end{figure}
The judgment $\wdOpSem(\mapEnv,\x,\sv)$ used in the side condition of rule \refToRule{invk} holds if $\WDUpdated{\x}{\emptyMap}$ holds, with $\mapEnvPrime{=}{\Update{\mapEnv}{\x}{\val}}$. The judgment $\WD{\sv}{\emptyMap}$ means well-definedness of a result. That is, restricting the domain of $\mapEnv$ to the variables reachable from $\sv$ (that is, either occurring in $\sv$, or, transitively, in values associated {with} reachable variables) we get a well-defined environment; thus, $\wdOpSem(\mapEnv,\x,\sv)$ holds if adding the association of $\sv$ {with} $\x$ preserves well-definedness of $\mapEnv$.

The additional argument $\map$ in the judgment $\WD{\sv}{\map}$ is a map from variables to {integer} numbers.  We write $\incrMap$ and $\decrMap$ for the maps $\{(\x,\map(\x)+1) \mid {\x\in\dom\map}\}$, and $\{(\x,\map(\x)-1) \mid \x\in\dom\map\}$, respectively.

In rule \refToRule{main}, this map is initially empty.
{In rule \refToRule{wd-var}, when a variable $\x$ defined in the environment is found the first time, it is added in the map with initial value $0$ before propagating the check to the associated value.
In rule \refToRule{wd-corec}, when it is found the second time, it is checked that constructors and right operands of interleave are traversed more times than tail operators, and if it is the case the variable is considered well-defined.
Rule \refToRule{wd-delay}, which is only added for the purpose of the soundness proof and should be not part of an implementation\footnote{{Indeed, it does not affect derivability, see \refToLemma{remove-delay} in the following.}}, performs the same check but then considers the variable occurrence as it is was the first, so that success of well-definedness is delayed. } \EZ{Note that rules \refToRule{wd-var}, \refToRule{wd-corec}, and \refToRule{wd-delay} can only be applied if $\x\in\dom{\mapEnv}$; in rule \refToRule{wd-corec}, this explicit side condition could be omitted since satisfied by construction of the proof tree.   }

In rule \refToRule{wd-fv}, a free variable is considered well-defined.\footnote{Non-well-definedness can only be detected on closed results.} 
In rules \refToRule{wd-cons} and \mbox{\refToRule{wd-tail}} the value associated {with} a variable is incremented/decremented by one, respectively, before propagating the check to the subterm. In rule \mbox{\refToRule{wd-nop}} the check is simply propagated to the subterms. In rule \refToRule{wd-$\IL$}, the check is also propagated to the subterms, but on the right-hand side the value associated with a variable is incremented by one before propagation; this reflects the fact that, in the worst case, $\At{\mapEnv}{\sv_1\IL\sv_2}{\indx}{\atfun_\mapEnv(\sv_1,\indx)}$, and this happens only for $\indx=0$, while for odd indexes $\indx$ we have that $\At{\mapEnv}{\sv_1\IL\sv_2}{\indx}{\atfun_\mapEnv(\sv_2,\indx-k)}$, with $k\geq 1$; more precisely, $k=1$ only when $\indx=1$; for all indexes $\indx>1$ (both even and odd), $k>1$. 
For instance, the example \lstinline!s() = 1:(s()!$\IL$\lstinline!s()^)!, which has the same semantics as \lstinline![1]!, would be considered not well-defined if we treated the interleaving as the pointwise arithmetic operators. 

\EZ{Note that the rules in \refToFigure{op-wd} can be immediately turned into an algorithm  which, given a stream value $\sv$, always terminates either successfully (finite proof tree), or with failure (no proof tree can be constructed). On the other hand, the rules in \refToFigure{stream-sem} defining the $\At{\mapEnv}{\sv}{\indx}{\nv}$ judgment can be turned into an algorithm which can possibly diverge (infinite proof tree). }

In the Appendix we show two examples of derivation of well-definedness and access to the $\indx$-th element: the result $\Caps{\x}{\{\x\mapsto \cons{0}{(\x\ [+]\ \y)},\y\mapsto \cons{1}{\y}\}}$ obtained by evaluating the call \lstinline!nat()!(\refToFigure{ex1}), and a more involved example, the result $\Caps{\x}{\{\x\mapsto \cons{0}{(\x\ [+]\ \y)\ ||\ (\x\ [+]\ \y), \y\mapsto 1:\y}\}}$ obtained by evaluating the call \lstinline{bfs_level()}  (\refToFigure{ex2}), with \lstinline{nat} and \lstinline{bfs_level} defined as  in \refToSection{examples}. Below we show an example of failing derivation:

\begin{figure}
\begin{small}
$$\begin{array}{l}
\infeRule{wd-var}
{\infeRule{wd-cons}
{\infeRule{wd-$\IL$}
{\infeRule{wd-corec}
{}
{\WD{\x}{\{\x\mapsto 1\}}}
\BigSpace
\infeRule{wd-tail}
{\infeRule{wd-tail}
{\infeRule{??}
{\mathtt{FAIL}}
{\WD{\x}{\{\x\mapsto 0\}}}}
{\WD{\tail{\x}}{\{\x\mapsto 1\}}}}
{\WD{\tail{\tail{\x}}}{\{\x\mapsto 2\}}}}
{\WD{\x\ ||\ \tail{\tail{\x}}}{\{\x\mapsto 1\}}}}
{\WD{0:(\x\ ||\ \tail{\tail{\x}})}{\{\x\mapsto 0\}}}}
{\WD{\x}{\emptyset}}
\end{array}$$
\end{small}
\caption{Failing derivation for ${\mapEnv=\{\x\mapsto 0:(\x\ ||\ \tail{\tail{\x}}) \}}$}
\label{fig:ex-non-wd}
\end{figure}

As depicted in \refToFigure{ex-non-wd}, the check succeeds for the left-hand component of the interleaving operator, while the proof tree cannot be completed for the other side. Indeed, the double application of the tail operator makes undefined access to stream elements with index greater than $1$, \DA{since the evaluation of $\atfun_\mapEnv(\x,2)$ infinitely triggers the evaluation of itself}.

To formally express and prove that well-definedness of a result implies termination of access to an arbitrary index, we introduce some definitions and notations.
First of all, since  the result is not relevant for the following technical treatment, for simplicity we will write $\AT{\sv}{\indx}$ rather than $\At{\mapEnv}{\sv}{\indx}{\nv}$. 
We call \emph{derivation} an either finite or infinite proof tree. We write $\WD{\sv'}{\map'}\premise\WD{\sv}{\map}$ to mean that $\WD{\sv'}{\map'}$ is a premise of a \mbox{(meta-)rule} where $\WD{\sv}{\map}$ is the conclusion, and $\premisestar$ for the reflexive and transitive closure of this relation. 

\begin{lemma}\label{lemma:basics}\
\begin{enumerate}
\item\label{i} A judgment $\WD{\sv}{\emptyset}$ has no derivation
iff the following condition holds:\\
\begin{tabular}{ll}
\refToRule{wd-stuck}&$\WD{\x}{\map'}\premisestar \WD{\mapEnv(\x)}{\Update{\map}{\x}{0}}\premise\WD{\x}{\map}\premisestar\WD{\sv}{\emptyMap}$\\
&for some $\x\in\dom{\mapEnv}$, 
%$\varset',\varset$, 
and $\map',\map$ s.t.\ $\x\not\in\dom{\map}, \map'(\x)\leq 0$.
\end{tabular}
\item\label{ii} {If the derivation of $\AT{\sv}{\jndx}$
is infinite, then }the following condition holds:\\
\begin{tabular}{ll}
\refToRule{at-$\infty$}&$\AT{\x}{\indx+k}\premisestar\AT{\mapEnv(\x)}{\indx}\premise\AT{\x}{\indx}\premisestar\AT{\sv}{\jndx}$\\
&for some $\x\in\dom{\mapEnv}$, 
%$\varset',\varset$, 
and $\indx, k \geq 0$.
\end{tabular}
\end{enumerate}
\end{lemma} 

{\begin{lemma}\label{lemma:second-occurrence} %\begin{enumerate}
%\item 
If $\AT{\x}{\indx'}\premisestar\AT{\sv'}{\indx}$, and $\WD{\sv'}{\map}\premisestar\WD{\sv}{\emptyset}$ with $\WD{\sv}{\emptyset}$ derivable, and $\x\in\dom{\map}$, then
\begin{quote}
 ${\WD{\x}{\map'}\premisestar\WD{\sv'}{\map}}$ for some $\map'$ such that $\map'(\x)-\map(\x) \leq \indx-\indx'$.
 \end{quote}
\end{lemma}}

\begin{proof}%[\refToLemma{second-occurrence}]
%\begin{description}
%\item[1$\Rightarrow 2$]
The proof is by induction on the {length} of the path in ${\AT{\x}{\indx'}\premisestar\AT{\sv'}{\indx}}$. 
\begin{description}
\item[Base] The {length} of the path is $0$, hence we have $\AT{\x}{\indx}\premisestar\AT{\x}{\indx}$. We also have  $\WD{\x}{\map}\premisestar\WD{\x}{\map}$, and we get the thesis since $\map(\x)=\map(\x)+\indx-\indx$.
\item[Inductive step] By cases on the rule applied to derive $\AT{\sv'}{\indx}$.
\begin{description}
\item[\refToRule{at-var}] {We have $\AT{\x}{\indx'}\premisestar\AT{\mapEnv(\y)}{\indx}\premise\AT{\y}{\indx}$. 
There are two cases:
\begin{itemize}
\item If $\y\not\in\dom{\map}$ (hence $\y\neq\x$), we have $\WD{\mapEnv(\y)}{\Update{\map}{\y}{0}}\premise\WD{\y}{\map}$ by rule \refToRule{wd-var}, the premise is derivable, hence by inductive hypothesis we have  ${\WD{\x}{\map'}\premisestar\WD{\mapEnv(\y)}{\Update{\map}{\y}{0}}}$, and ${\map'(\x)\leq\Update{\map}{\y}{0}(\x)+\indx-\indx'}=\map(\x)+\indx-\indx'$, hence we get the thesis.
\item If $\y\in\dom{\map}$, then it is necessarily $\map(\y)>0$, \EZ{otherwise, by \refToLemma{basics}-(1), $\WD{\sv}{\emptyset}$ would not be derivable}. Hence, we have \linebreak ${\WD{\mapEnv(\y)}{\Update{\map}{\y}{0}}\premise\WD{\y}{\map}}$ by rule \refToRule{wd-delay}, hence by inductive hypothesis we have  ${\WD{\x}{\map'}\premisestar\WD{\mapEnv(\y)}{\Update{\map}{\y}{0}}}$, and ${\map'(\x)\leq\Update{\map}{\y}{0}(\x)+\indx-\indx'}$. There are two subcases:
\begin{itemize}
\item If $\y\neq\x$, then $\Update{\map}{\y}{0}(\x)=\map(\x)$, and we get the thesis as in the previous case.
\item If $\y=\x$, then $\Update{\map}{\x}{0}(\x)=0$, hence ${\map'(\x)\leq\indx-\indx'\leq\map(\x)+\indx-\indx'}$, since $\map(\x)>0$.
\end{itemize}
\end{itemize}}
\item[\refToRule{at-cons-0}] Empty case, since the derivation for $\AT{\cons{\nv}{\sv}}{0}$ does not contain a node $\AT{\x}{\indx'}$.
\item[\refToRule{at-cons{-succ}}] We have $\AT{\cons{\nv}{\sv}}{\indx}$, and $\AT{\x}{\indx'}\premisestar\AT{\sv}{\indx-1}$. 
Moreover, we can derive $\WD{\cons{\nv}{\sv}}{\map}$ by rule \refToRule{wd-cons}, and by inductive hypothesis we also have  $\WD{\x}{\map'}\premisestar\WD{\sv}{\incrMap}$, with $\map'(\x)\leq\incrMap(\x)+(\indx-1)-\indx'$, hence we get the thesis.
\item[\refToRule{at-tail}] This case is symmetric to the previous one.
\item[\refToRule{at-nop}] We have $\AT{\PW{\sv_1}{\pwop}{\sv_2}}{\indx}$, and either $\AT{\x}{\indx'}\premisestar\AT{\sv_1}{\indx}$, or \linebreak ${\AT{\x}{\indx'}\premisestar\AT{\sv_2}{\indx}}$. Assume the first case holds, the other is analogous.
Moreover, we can derive $\WD{\PW{\sv_1}{\pwop}{\sv_2}}{\map}$ by rule \refToRule{wd-nop}, and by inductive hypothesis we also have  $\WD{\x}{\map'}\premisestar\WD{\sv_1}{\map}$, with $\map'(\x)\leq\map(\x)+\indx-\indx'$, hence we get the thesis.

\item[\refToRule{{at-}$\IL$-even}] We have $\AT{\sv_1\IL\sv_2}{2\indx}$ and $\AT{\x}{\indx'}\premisestar\AT{\sv_1}{\indx}$. By inductive hypothesis, we have $\WD{\x}{\map'}\premisestar\WD{\sv_1}{\map}$, with $\map'(\x)\leq \map(\x)+\indx-\indx'$. Moreover, $\WD{\sv_1}{\map}\premise{}\WD{\PW{\sv_1}{\IL}{\sv_2}}{\map}$ holds by rule \mbox{\refToRule{wd-$\IL$}}, hence we have $\WD{\x}{\map'}\premisestar\WD{\PW{\sv_1}{\IL}{\sv_2}}{\map}$ with ${\map'(\x)\leq \map(\x)+2\indx-\indx'}$ and, thus, the thesis. 

\item[\refToRule{{at-}$\IL$-odd}] We have $\AT{\sv_1\IL\sv_2}{2\indx+1}$ and $\AT{\x}{\indx'}\premisestar\AT{\sv_2}{\indx}$. By inductive hypothesis, we have $\WD{\x}{\map'}\premisestar\WD{\sv_2}{\incrMap}$, with $\map'(\x)\leq \incrMap(\x)+\indx-\indx'$. Moreover, $\WD{\sv_2}{\map}\premise{}\WD{\PW{\sv_1}{\IL}{\sv_2}}{\map}$ holds by rule \mbox{\refToRule{wd-$\IL$}}, hence we have $\WD{\x}{\map'}\premisestar\WD{\PW{\sv_1}{\IL}{\sv_2}}{\map}$ with $\map'(\x)\leq \map(\x)+2\indx+1-\indx'$ and, thus, the thesis.
\end{description}
\end{description}
\end{proof}

{\begin{lemma}\label{lemma:first-occurrence} 
If $\AT{\x}{\indx'}\premisestar\AT{\sv}{\indx}$, and $\WD{\sv}{\emptyset}$ derivable, then 
\begin{quote}
${\WD{\x}{\map}\premisestar\WD{\sv}{\emptyset}}$ for some $\map$ such that $\x\not\in\dom{\map}$.
\end{quote}
\end{lemma}}
\begin{proof}%[\refToLemma{first-occurrence}]
Easy variant of the proof of \refToLemma{second-occurrence}.
\end{proof}

\begin{theorem}\label{theo:iff}
{If $\WD{\sv}{\emptyMap}$ has a derivation then, for all $\jndx$, $\AT{\sv}{\jndx}$ either has no derivation or a finite derivation.}\EZComm{$\AT{\sv}{\jndx}$ potrebbe fallire in caso di ambiente aperto}
\end{theorem}

\begin{proof}
{Assume by contradiction that $\AT{\sv}{\jndx}$ has an infinite derivation for some $\jndx$, and $\WD{\sv}{\emptyMap}$ is derivable.}
%\begin{description}
%\item[$\Rightarrow$] 
By \refToLemma{basics}-({\ref{ii}}), the following condition holds:
\begin{quote}
\begin{tabular}{ll}
\refToRule{at-$\infty$}&$\AT{\x}{\indx+k}\premisestar\AT{\mapEnv(\x)}{\indx}\premise\AT{\x}{\indx}\premisestar\AT{\sv}{\jndx}$\\
&for some $\x\in\dom{\mapEnv}$, 
%$\varset',\varset$, 
and $\indx, k \geq 0$.
\end{tabular}
\end{quote}
Then, starting from the right, by \refToLemma{first-occurrence} we have ${\WD{\x}{\map}\premisestar\WD{\sv}{\emptyMap}}$ for some $\map$ such that $\x\not\in\dom{\map}$; by rule \refToRule{wd-var} ${\WD{\mapEnv(\x)}{\Update{\map}{\x}{0}}\premise\WD{\x}{\map}}$, and finally by \refToLemma{second-occurrence} we have:
\begin{quote}
\begin{tabular}{ll}
\refToRule{wd-stuck}&$\WD{\x}{\map'}\premisestar \WD{\mapEnv(\x)}{\Update{\map}{\x}{0}}\premise\WD{\x}{\map}\premisestar\WD{\sv}{\emptyMap}$\\
&for some $\x\in\dom{\mapEnv}$, 
%$\varset',\varset$, 
and $\map',\map$ s.t.\ $\x\not\in\dom{\map}, \map'(\x){\leq {- k}}\leq 0$.
\end{tabular}
\end{quote}hence this is absurd {by \refToLemma{basics}-(\ref{i})}.
\end{proof}

\section{An optimized algorithm for well-definedness}\label{sect:opt-wd}

{The  definition of well-definedness in \refToFigure{op-wd} can be easily turned into an algorithm, since, omitting rule \refToRule{wd-delay}, at each point in the derivation there is at most one applicable rule.} Now we will discuss its time complexity, assuming that insertion, update and lookup are performed in constant time. 
It is easy to see that when we find a stream constructor we need to perform an update of the map $\mapEnv$ for every variable
  in its domain. If we  consider the following environment:
 \begin{small}
 \begin{quote}
$\mapEnv=(x_0,\{x_0\mapsto 0:x_1,x_1\mapsto 0:x_2, x_2\mapsto 0:x_3, x_3\mapsto 0:x_4,\ \ldots\ , x_n\mapsto 0:x_0\})$
\end{quote}
\end{small}
we get the derivation presented in \refToFigure{wd_worst_case}.
\begin{figure}
\begin{small}
$$\begin{array}{l}
\infeRule{wd-var}
{\infeRule{wd-cons}
{\infeRule{wd-var}
{\infeRule{wd-cons}
{\infeRule{wd-var}
{\infeRule{wd-cons}
{\infeRule{wd-var}
{\vdots}
{\WD{\x_3}{\{x_0\mapsto 3, x_1\mapsto 2, x_2\mapsto 1 \}}}}
{\WD{0:\x_3}{\{x_0\mapsto 2, x_1\mapsto 1, x_2\mapsto 0 \}}}}
{\WD{\x_2}{\{x_0\mapsto 2, x_1\mapsto 1\}}}}
{\WD{0:\x_2}{\{x_0\mapsto 1, x_1\mapsto 0\}}}}
{\WD{\x_1}{\{x_0\mapsto 1\}}}}
{\WD{0:\x_1}{\{x_0\mapsto 0\}}}}
{\WD{\x_0}{\emptyset}}
\end{array}$$
\end{small}
\caption{}\label{fig:wd_worst_case}
\end{figure}
Here, the number of constructor occurrences for which we have to perform an update of all variables in the domain of the map is
  linearly proportional to the number $N$ of nodes in the derivation tree; since the domain is increased by one for each new variable, and the total number of variables is again linearly proportional to $N$, it is easy to see that we have a time complexity quadratic in $N$.

We propose now an optimized version of the well-definedness check, having a time complexity of $O(N \log N)$.  
On the other hand, the version we provided in \refToFigure{op-wd} is more abstract, hence more convenient for the proof of \refToTheorem{iff}.

In the optimized version, given in \refToFigure{opt-wd}, the judgment has shape $\owd{\sv}{\optmap}{\p}$, where $\p$ represents a path in the proof tree where each element corresponds to a visit of either the constructor or the right operand of interleave (value 1 for both) or the tail operator (value -1), and $\optmap$ associates with each variable an index (starting from 0) corresponding to the point in the path $\p$ where
the variable was found the first time. The only operation performed on a path $\p$ is the addition $\p\appOp b$ of an element $b$ at the end.
%%represented by the notations $\p\appOp 1$ and $\p\appOp (-1)$, respectively. 

\begin{figure}[h]
\begin{small}
\begin{grammatica}
\produzione{\optmap}{{\x_1\mapsto\indx_1 \ldots \x_n\mapsto\indx_k} \Space (\indx\geq0)}{map from variables to indexes}\\
\produzione{\p}{b_1 b_2 \ldots b_n}{sequence of either 1 or -1}
\\[2ex]
\end{grammatica}
\\
\hrule 
$\begin{array}{l}
\\
{  
\NamedRule{main}
{\owdUpdated{\x}{\emptyMap}{\emptyPath}}
{\wdOpSem(\mapEnv,\x,\val)}
{\mapEnv'=\Update{\mapEnv}{\x}{\val}}
}
\BigSpace  
\NamedRule{owd-var}
{\owd{\mapEnv(\x)}{\Update{\optmap}{\x}{\indx}}{\p}}
{\owd{\x}{\optmap}{\p} }
{\x\not\in\dom\optmap\\
{\indx}=\len{\p}}
\\[5ex]
\NamedRule{owd-corec}
{}
{\owd{\x}{\optmap}{\p}}
{\x\in\dom\optmap\\
{\sumFrom(\optmap(\x),\p)>0}}
\BigSpace
{\NamedRule{owd-fv}
{}
{\owd{\x}{\optmap}{\p} }
{\x\not\in\dom\mapEnv}}
\\[5ex]
\NamedRule{owd-cons}
{\owd{\sv}{\optmap}{{\p\appOp 1}}}
{\owd{n:\sv}{\optmap}{\p}}
{}
\BigSpace
\NamedRule{owd-tail}
{\owd{\sv}{\optmap}{{\p\appOp (-1)}}}
{\owd{\tail{\sv}}{\optmap}{\p}}
{}
\\[5ex]
\NamedRule{owd-nop}
{\owd{\sv_1}{\optmap}{\p}\Space\owd{\sv_2}{\optmap}{\p}}
{\owd{\sv_1\pwnop\sv_2}{\optmap}{\p}}
{}
\BigSpace
{
\NamedRule{owd-$\IL$}
{\owd{\sv_1}{\optmap}{\p}\Space\owd{\sv_2}{\optmap}{\p\appOp 1}}
{\owd{\sv_1\IL\sv_2}{\optmap}{\p}}
{}
}
\\[5ex]
{\NamedRule{sum-0}{\sumFrom(\p)=n}{\sumFrom(0,\p)=n}{}}
\BigSpace
{\NamedRule{sum-n}{\sumFrom(n-1,{b_2 \ldots b_n})=n'}{\sumFrom(n,{b_1 b_2 \ldots b_n})=n'}{n>0}}
\\[5ex]
{\NamedRule{sum-b}{}{\sumFrom(\emptyPath)=0}{}}
\BigSpace
{\NamedRule{sum-i}{\sumFrom({b_2 \ldots b_n})=n}{\sumFrom({b_1 b_2 \ldots b_n})={b_1}+n}{}}
\end{array}$
\end{small}
\caption{Optimized operational definition of well-definedness}\label{fig:opt-wd}
\end{figure}

In rule \refToRule{main}, both the map and the path are initially empty.
In rule \mbox{\refToRule{owd-var}}, a variable $\x$ defined in the environment, found for the first time, is added in the map with as index the length of the current path.
In rule \refToRule{owd-corec}, when the same variable is found the second time, the auxiliary function $\sumFrom$ checks that more constructors and right operands of interleave have been traversed than tail operators {(see below)}. In rule \refToRule{owd-fv}, a free variable is considered well-defined
{as in the corresponding rule in \refToFigure{op-wd}}. In rules \refToRule{owd-cons}, \refToRule{owd-tail} and \refToRule{op-wd}, the value corresponding to the traversed operator is added {at the end of} the path  (1 for the constructor and the right operand of interleave, -1 for the tail operator). Lastly, rules \refToRule{owd-nop} behaves in a similar way as in \refToFigure{op-wd}.
The semantics of the auxiliary function $\sumFrom$ is straightforward: starting from the point in the path where the variable was found the first time, the sum of all the elements is returned.  %%{to} make sure that the result is positive. 
%{Similarly as happens for the definition in \refToFigure{op-wd},
%this check ensures that more constructors than tail operators have been traversed while inspecting a certain variable}.   

Let us now consider again the example above:
\begin{small}
\begin{quote}
$
\mapEnv=(x_0,\{x_0\mapsto 0:x_1,x_1\mapsto 0:x_2, x_2\mapsto 0:x_3, x_3\mapsto 0:x_4,\ \ldots\ , x_n\mapsto 0:x_0\})
$
\end{quote}
\end{small}
By the new predicate $\mathsf{owd}$, we get a derivation tree of the same shape as in \refToFigure{wd_worst_case}. However, 
$\sumFrom$ is applied to the path $\p$ only at the leaves, and the length of $\p$ is linearly proportional to the depth of the
  derivation tree, which coincides with the number $N$ of nodes in this specific case; hence, the time complexity
  to compute $\sumFrom(0,\p)$ (that is, $\sumFrom(\optmap(x_0),\p)$) is linear in $N$. Finally, since for inner nodes only constant time operations are performed\footnote{This holds for any valid derivation tree and not for this specific case.} (addition at the end of the path, and map insertion and lookup), the overall time complexity is \mbox{linear in $N$}.

As worst case in terms of time complexity for the predicate $\mathsf{owd}$, consider
\begin{small}
\begin{quote}
$
{\mapEnv_\indx=(\x_0,\{\x_0\mapsto 0:\x_1[+]\x_1,\x_1\mapsto 0:\x_2[+]\x_2, \x_2\mapsto 0:\x_3[+]\x_3, \ldots\ , x_\indx\mapsto 0:x_0\})}
$
\end{quote}
\end{small}
The derivation tree for this environment is shown in \refToFigure{owd_worst_case}, where $\optmap_i$ abbreviates the map $\{\x_0\mapsto 0, \x_1\mapsto 1, \ldots, \x_\indx\mapsto \indx\}$.

\begin{figure}
\begin{scriptsize}
$\begin{array}{l}
\infeRule{owd-var}
{\infeRule{owd-cons}
{\infeRule{owd-nop}
{\infeRule{owd-var}
{\infeRule{owd-cons}
{\infeRule{owd-nop}
{\infeRule{owd-var}
{\infeRule{owd-cons}
{\infeRule{owd-nop}
{\vdots}
{\owd{\x_3[+]\x_3}{\optmap_2}{[1,1,1]}}
}
{\owd{0:\x_3[+]\x_3}{\optmap_2}{[1,1]}}
}
{\owd{\x_2}{\optmap_1}{[1,1]}}
&
\infeRule{owd-var}
{\vdots
}
{\owd{\x_2}{\optmap_1}{[1,1]}}}
{\owd{\x_2[+]\x_2}{\optmap_1}{[1,1]}}
}
{\owd{0:\x_2[+]\x_2}{\optmap_1}{[1]}}
}
{\owd{\x_1}{\optmap_0}{[1]}}
&
\infeRule{owd-var}{\vdots}
{\owd{\x_1}{\optmap_0}{[1]}}}
{\owd{\x_1[+]\x_1}{\optmap_0}{[1]}}
}
{\owd{0:\x_1[+]\x_1}{\optmap_0}{\epsilon}}
}
{\owd{\x_0}{\emptyset}{\epsilon}}
\end{array}$
\end{scriptsize}
\caption{}\label{fig:owd_worst_case}
\end{figure}
As {already noticed, for inner nodes only constant time operations are performed, and the length of the paths in the leaves is linearly proportional to the depth $D$ of the derivation tree; however, in this worst case the number of leaves is not just one, but is linearly proportional to the
    total number $N$ of nodes in the derivation tree, hence the depth $D$ is linearly proportional to $\log N$.
    Therefore the overall time complexity is $O(N \cdot D)$, that is, $O(N \cdot \log N)$}. 

%%%%%%%%%%%%%%%  OPTIONAL EXAMPLE  %%%%%%%%%%%%%%%%%
%
%\begin{comment}
%\begin{figure}
%\begin{small}
%$\begin{array}{l}
%\NamedRule{owd-var}
%{\NamedRule{owd-cons}
%{\NamedRule{{owd-nop}}
%{\NamedRule{owd-corec}
%{}
%{\owd{\x}{\{\x\mapsto0\}}{[1]}}
%{}
%\BigSpace
%\NamedRule{owd-var}
%{\NamedRule{owd-cons}
%{\NamedRule{owd-corec}
%{}
%{\owd{\y}{\{\x\mapsto0,\y\mapsto 1\}}{[1,1]}}
%{}}
%{\owd{\cons{1}{\y}}{\{\x\mapsto0,\y\mapsto1\}}{[1]}}
%{}}
%{\owd{\y}{\{\x\mapsto0\}}{[1]}}
%{}}
%{\owd{\x\ [+]\ \y}{\{\x\mapsto0\}}{[1]}}
%{}}
%{\owd{\cons{0}{(\x\ [+]\ \y)}}{\{\x\mapsto0\}}{\emptyPath}}
%{}}
%{\owd{\x}{\emptyMap}{\emptyPath}}
%{}
%\end{array}$
%\end{small}
%\caption{Derivation for $\mapEnv=\Caps{\x}{\{\x\mapsto \cons{0}{(\x\ [+]\ \y)},\ y\mapsto \cons{1}{\y}\}}$.}\label{fig:ex1_owd}
%\end{figure}
%\end{comment}
%
%%%%%%%%%%%%%%%  END OPTIONAL EXAMPLE  %%%%%%%%%%%%%%%%%

We now show that the optimized version of the judgment has the same semantics as its counterpart presented in \refToSection{wd}.
{First of all we formally state that, in \refToFigure{op-wd}, rule \refToRule{wd-delay} does not affect derivability.
\begin{lemma}\label{lemma:remove-delay}
A judgment $\WD{\sv}{\emptyset}$ has a derivation iff it has a derivation which does not use rule \refToRule{wd-delay}.
\end{lemma}
\begin{proof}
The right-to-left implication is obvious. If $\WD{\sv}{\emptyset}$ uses rule \refToRule{wd-delay}, all the (first in their path) occurrences of the rule can be replaced by rule \refToRule{wd-corec}, still getting a derivation. 
\end{proof}
}

{Then, we define a relation between the auxiliary structures used in the two judgments:
\begin{quote}\label{def-env-rel} 
For all $\map$ and $\Pair{\optmap}{\p}$, $\wdRel{\map}{\optmap}{\p}$ holds iff\\
 $\dom\map=\dom{\optmap}$ and, for all $\x\in\dom\map$, $\map(\x)=\sumFrom(\optmap(x),\p)$.
\end{quote}}
In this way, we have the following generalization, whose straightforward proof is in the Appendix.
\begin{theorem}\label{theo:wd-eq-gen} 
If $\wdRel{\optmap}{\optmap}{\p}$, then, for all $\sv$, $\WD{\sv}{\map}$ is derivable iff $\owd{\sv}{\optmap}{\p}$ is derivable.
\end{theorem}

\begin{corollary}\label{theo:wd-eq} 
$\WD{\sv}{\emptyMap}$ is derivable iff $\owd{\sv}{\emptyMap}{\emptyPath}$ is derivable.
\end{corollary}

\section{Related and future work}\label{sect:conclu}
As  mentioned in \refToSection{intro}, our approach extends regular corecursion, which originated from \emph{co-SLD resolution} \cite{Simon06,SimonBMG07,Ancona13,AnconaDovier15}, where already considered goals (up to unification), called \emph{coinductive hypotheses}, are successfully solved. Language constructs that support this programming style have also been proposed in the functional \cite{Jeannin17} and object-oriented \cite{AnconaZ12,AnconaBDZ20} paradigm. 

There have been a few attempts of extending the expressive power of regular corecursion. Notably, \emph{structural resolution} \cite{KomendantskayaJS16,KomendantskayaPS16} is an operational semantics for logic programming where infinite derivations that cannot be built in finite time are generated lazily, and only partial answers are shown. Another approach is the work in \cite{Courcelle83}, introducing algebraic trees and equations as generalizations of regular ones. \EZ{Such proposals share, even though with different techniques and in a different context,  our aim of extending regular corecursion; on the other hand, the fact that corecursion is \emph{checked} is, at our knowledge, a novelty of our work.}

\EZComm{ripescato da ICTCS}
\EZ{For the operators considered in the calculus and some examples, our main sources of inspiration have been the works of
Rutten \cite{Rutten05}, where a coinductive calculus of streams of real numbers is defined,
and Hinze \cite{Hinze10}, where a calculus of generic streams is defined in a constructive way and implemented in Haskell.}

 \EZ{In this paper, as in all the above mentioned approaches derived from co-SLD resolution, the aim is to provide an \emph{operational} semantics, designed to directly lead to an implementation.
 That is, even though streams are infinite objects (terms where the constructor is the only operator, defined coinductively), evaluation handles \emph{finite} representations, and is defined by an \emph{inductive} inference system. Coinductive approaches can be adopted to obtain more abstract semantics of calculi with infinite terms. For instance, \cite{Czajka20} defines a coinductive semantics of the infinitary lambda-calculus where, roughly, the semantics of terms with an infinite reduction sequence is the infinite term obtained as limit. In coinductive logic programming, co-SLD resolution is the operational counterpart of a coinductive semantics where a program denotes a set of infinite terms. In \cite{AnconaBDZ20}, analogously, regular corecursion is shown to be sound with respect to an abstract coinductive semantics using \emph{flexible coinduction} \cite{AnconaDZ@esop17,Dagnino21}, see below.}

Our calculus is an enhancement of that presented in \cite{AnconaBZ21}, with two main significant contributions: (1) the interleaving operator,
challenging since it is based on a non-trivial recursion schema; (2) an optimized definition} of the {runtime} well-definedness check, as {a useful} basis for an implementation.
  \EZComm{ripristinato}
\EZ{Our main technical results are \refToTheorem{iff}, stating that passing the runtime well-definedness check performed for a function call {prevents} non-termination in accessing elements in the resulting stream, and \refToTheorem{wd-eq-gen}, stating that the optimized version is equivalent.}

{Whereas in \cite{AnconaBZ21} the well-definedness check was also a necessary condition to guarantee termination, this is not the case here, due to the interleaving operator.} Consider,  for instance, the following example:
  $\mapEnv=\{ s\mapsto (\tail{s} \IL s) \IL 0{:}s \}$. The judgment $\WD{s}{\emptyset}$ is not derivable, in particular because of $\tail{s}$, since
  ${\WD{s}{\{s\mapsto -1\}}}$ is not derivable and, hence, $\WD{\tail{s}}{\{s\mapsto 0\}}$, $\WD{\tail{s}\IL s}{\{s\mapsto 0\}}$, and \linebreak ${\WD{(\tail{s}\IL s)\IL 0{:}s}{\{s\mapsto 0\}}}$. However, $\Atop{\mapEnv}{s}{i}$ is well-defined for all indexes $i$; indeed, $\At{\mapEnv}{s}{1}{0}$  is derivable, $\At{\mapEnv}{s}{0}{k}$ is derivable iff $\At{\mapEnv}{s}{1}{k}$ is derivable, and, for all $i>1$, $\At{\mapEnv}{s}{i}{k}$ is derivable iff
  $\At{\mapEnv}{s}{j}{k}$ is derivable for some $j<i$, hence $\At{\mapEnv}{s}{i}{0}$ is derivable for all $i$.
  We leave for future work the investigation of a complete check.} 
 
In future work, we plan to also prove soundness of the operational well-definedness with respect to its abstract definition. Completeness does not hold, as shown by the example
\lstinline!zeros() = [0] [*] zeros()! \DAComm{Counter examples for future work}
which is not well-formed operationally, but admits as unique solution the stream of all zeros.

Finally, in the presented calculus a cyclic call is detected by rule \refToRule{corec} if it is syntactically \EZ{the same of} some in the call trace.
Although such a rule allows cycle detection  for all the examples \EZComm{vero?} presented in this paper, it is not complete with respect to the abstract
notion where expressions denoting the same stream are equivalent, as illustrated by the following alternative definition of function \lstinline{incr}
as presented in \refToSection{examples}:
\begin{lstlisting} 
incr_reg(s) = (s(0)+1):incr_reg(s^) 
\end{lstlisting}
If syntactic equivalence is used to detect cycles, then the call \lstinline{incr_reg([0])} diverges, since the terms passed as argument to the recursive calls are all syntactically different; as an example, consider the arguments $\x$ and $\tail{x}$ passed to the initial call and to the first recursive call, respectively, in the environment $\mapEnv=\{\x\mapsto \cons{0}{\x}\}$; they are syntactically different, but denote the same stream.

%% \begin{lstlisting}
%% first(s) = s(0):first(s) // works with syntactic equivalence
%% first2(s) = s(0):first2(s(0):s^) // does not work
%% \end{lstlisting}
\EZComm{cut: Indeed, we get an infinite derivation for \lstinline!first2([1])!
  with call traces of increasing shape
  ${\{\mbox{\lstinline!first2($x$)!}\mapsto \y_1,\mbox{\lstinline!first2(1:$x$^)!}\mapsto \y_2,\mbox{\lstinline!first2(1:(1:$x$^)^)!}\mapsto \y_3,\ldots\}}$ in the environment $\{\x\mapsto\cons{1}{\x}\}$.}
In future work we plan to investigate more expressive operational characterizations of equivalence.

Other interesting directions for future work are the following.
\begin{itemize}
\item  Investigate additional operators and the expressive power of the calculus.\DAComm{Conjecture regarding the expressive power of the calculus: (1) more expressive than polynomial streams (easy to prove, by exploiting that the fib function is not definable with a polynomial) (2) not possible to extract the stream of elements at even index (to be proved).}
\item  Design a static type system to prevent runtime errors such as the non-well-definedness of a stream.
\item Extend corecursive definition to \emph{flexible} corecursive definitions \cite{Dagnino21,DagninoAZ20} where programmers can define specific behaviour when a cycle is detected.
\EZ{In this way we could get termination in cases where lazy evaluation diverges. For instance, assuming to allow also booleans results for functions,
we could define the predicate \lstinline{allPos}, checking that all the elements of a stream are positive, specifying as result \lstinline{true} when a cycle is detected; in this way, e.g., \lstinline{allPos(one_two)} would return the correct result. }
\end{itemize}

\newpage

% ---- Bibliography ----

\bibliographystyle{plain}
\bibliography{bib}

\begin{thebibliography}{10}

\bibitem{Ancona13}
Davide Ancona.
\newblock Regular corecursion in {P}rolog.
\newblock {\em Computer Languages, Systems {\&} Structures}, 39(4):142--162,
  2013.

\bibitem{AnconaBDZ20}
Davide Ancona, Pietro Barbieri, Francesco Dagnino, and Elena Zucca.
\newblock Sound regular corecursion in {coFJ}.
\newblock In Robert Hirschfeld and Tobias Pape, editors, {\em ECOOP'20 -
  Object-Oriented Programming}, volume 166 of {\em LIPIcs}, pages 1:1--1:28.
  Schloss Dagstuhl - Leibniz-Zentrum f{\"u}r Informatik, 2020.

\bibitem{AnconaBZ21}
Davide Ancona, Pietro Barbieri, and Elena Zucca.
\newblock Enhanced regular corecursion for data streams.
\newblock In {\em ICTCS'21 - Italian Conf. on Theoretical Computer Science},
  2021.

\bibitem{AnconaDZ@esop17}
Davide Ancona, Francesco Dagnino, and Elena Zucca.
\newblock Generalizing inference systems by coaxioms.
\newblock In Hongseok Yang, editor, {\em 26th European Symposium on
  Programming, {ESOP} 2017}, volume 10201 of {\em Lecture Notes in Computer
  Science}, pages 29--55, Berlin, 2017. Springer.

\bibitem{AnconaDovier15}
Davide Ancona and Agostino Dovier.
\newblock A theoretical perspective of coinductive logic programming.
\newblock {\em Fundamenta Informaticae}, 140(3-4):221--246, 2015.

\bibitem{AnconaZ12}
Davide Ancona and Elena Zucca.
\newblock Corecursive {F}eatherweight {J}ava.
\newblock In {\em FTfJP'12 - Formal Techniques for Java-like Programs}, pages
  3--10. ACM Press, 2012.

\bibitem{Courcelle83}
Bruno Courcelle.
\newblock Fundamental properties of infinite trees.
\newblock {\em Theoretical Computer Science}, 25:95--169, 1983.

\bibitem{Czajka20}
Lukasz Czajka.
\newblock A new coinductive confluence proof for infinitary lambda calculus.
\newblock {\em Logical Methods in Computer Science}, 16(1), 2020.

\bibitem{Dagnino21}
Francesco Dagnino.
\newblock {\em Flexible Coinduction}.
\newblock PhD thesis, DIBRIS, University of Genova, 2021.

\bibitem{DagninoAZ20}
Francesco Dagnino, Davide Ancona, and Elena Zucca.
\newblock Flexible coinductive logic programming.
\newblock {\em Theory and Practice of Logic Programming}, 20(6):818--833, 2020.
\newblock Issue for ICLP 2020.

\bibitem{Hinze10}
Ralf Hinze.
\newblock Concrete stream calculus: An extended study.
\newblock {\em Journal of Functional Programming}, 20(5–6):463–535, 2010.

\bibitem{JeanninK12}
Jean{-}Baptiste Jeannin and Dexter Kozen.
\newblock Computing with capsules.
\newblock {\em Journal of Automata, Languages and Combinatorics},
  17(2-4):185--204, 2012.

\bibitem{Jeannin17}
Jean-Baptiste Jeannin, Dexter Kozen, and Alexandra Silva.
\newblock {CoCaml}: Functional programming with regular coinductive types.
\newblock {\em Fundamenta Informaticae}, 150:347--377, 2017.

\bibitem{KomendantskayaJS16}
Ekaterina Komendantskaya, Patricia Johann, and Martin Schmidt.
\newblock A productivity checker for logic programming.
\newblock In Manuel~V. Hermenegildo and Pedro L{\'{o}}pez{-}Garc{\'{\i}}a,
  editors, {\em Logic-Based Program Synthesis and Transformation - {LOPSTR}
  2016, Revised Selected Papers}, volume 10184 of {\em Lecture Notes in
  Computer Science}, pages 168--186. Springer, 2016.

\bibitem{KomendantskayaPS16}
Ekaterina Komendantskaya, John Power, and Martin Schmidt.
\newblock Coalgebraic logic programming: from semantics to implementation.
\newblock {\em J. Log. Comput.}, 26(2):745--783, 2016.

\bibitem{Rutten05}
Jan J. M.~M. Rutten.
\newblock A coinductive calculus of streams.
\newblock {\em Mathematical Structures in Computer Science}, 15(1):93--147,
  2005.

\bibitem{Simon06}
Luke Simon.
\newblock {\em Extending logic programming with coinduction}.
\newblock PhD thesis, University of Texas at Dallas, 2006.

\bibitem{SimonBMG07}
Luke Simon, Ajay Bansal, Ajay Mallya, and Gopal Gupta.
\newblock Co-logic programming: Extending logic programming with coinduction.
\newblock In Lars Arge, Christian Cachin, Tomasz Jurdzinski, and Andrzej
  Tarlecki, editors, {\em Automata, Languages and Programming, 34th
  International Colloquium, {ICALP} 2007}, volume 4596 of {\em Lecture Notes in
  Computer Science}, pages 472--483. Springer, 2007.

\end{thebibliography}

% ---- Appendix ----

\appendix
\section{Examples of derivations}

\begin{figure}
\begin{small}
$$\begin{array}{l}
\f()=\g()\\
g()=1:\f()
\\[6ex]
\infeRule{invk}
{\infeRule{invk}
{\infeRule{cons}
{\infeRule{val}
{}
{\opsem{1}{\emptyset}{\{\f()\mapsto\x,\g()\mapsto\y\}}{1}{\emptyset}}
&
\infeRule{corec}
{}
{\opsem{\f()}{\emptyset}{\{\f()\mapsto\x,\g()\mapsto\y\}}{\x}{\emptyset}}}
{\opsem{1:\f()}{\emptyset}{\{\f()\mapsto\x,\g()\mapsto\y\}}{1:\x}{\emptyset}}}
{\opsem{\g()}{\emptyset}{\{\f()\mapsto\x\}}{\y}{\{\y\mapsto1:\x\}}}}
{\opsem{\f()}{\emptyset}{\emptyset}{\x}{\{\x\mapsto\y,\y\mapsto1:\x\}}}
\end{array}$$
\end{small}
\caption{Example of derivation}\label{fig:derivation1}
\end{figure}

\begin{figure}
\begin{small}
$$\begin{array}{l}
\f()=\g(2:\f())\\
g(s)=1:s
\\[6ex]
\infeRule{invk}
{\infeRule{args}
{\infeRule{cons}
{\infeRule{val}
{}
{\opsem{2}{\emptyset}{\{\f()\mapsto\x\}}{2}{\emptyset}}
&
\infeRule{corec}
{}
{\opsem{\f()}{\emptyset}{\{\f()\mapsto\x\}}{\x}{\emptyset}}}
{\opsem{2:\f()}{\emptyset}{\{\f()\mapsto\x\}}{2:\x}{\emptyset}} & \DA{T_1}}
{\opsem{\g(2:\f())}{\emptyset}{\{\f()\mapsto\x\}}{\y}{\{\y\mapsto1:2:\x\}}}
}
{\opsem{\f()}{\emptyset}{\emptyset}{\x}{\{\x\mapsto\y,\y\mapsto1:2:\x\}}}
\\[8ex]
T_1=
\infeRule{invk}
{\infeRule{val}
{}
{\opsem{1:2:\x}{\emptyset}{\{\g(2:\x)\mapsto\y,\f()\mapsto\x\}}{1:2:\x}{\emptyset}}}
{\opsem{\g(2:\x)}{\emptyset}{\{\f()\mapsto\x\}}{\y}{\{\y\mapsto1:2:\x\}}}
\end{array}$$
\end{small}
\caption{Example of derivation}\label{fig:derivation2}
\end{figure}

\begin{figure}
\begin{small}
$$\begin{array}{l}
\undef()=(\undef()(0)):\undef()
\\[6ex]
\infeRule{invk}
{\infeRule{cons}
{T_1
&
\infeRule{corec}
{}
{\opsem{\undef()}{\emptyset}{\{\undef()\mapsto\x\}}{\x}{\emptyset}}}
{\opsem{(\undef()(0)):\undef()}{\emptyset}{\{\undef()\mapsto\x\}}{?}{?}}
}
{\opsem{\undef()}{\emptyset}{\emptyset}{?}{?}}
{}
\\[8ex]
T_1=
\infeRule{at}
{\infeRule{corec}
{}
{\opsem{\undef()}{\emptyset}{\{\undef()\mapsto\x\}}{\x}{\emptyset}}
&
\infeRule{val}
{}
{\opsem{0}{\emptyset}{\{\undef()\mapsto\x\}}{0}{\emptyset}}}
{\opsem{\undef()(0)}{\emptyset}{\{\undef()\mapsto\x\}}{?}{?}}
%{\At{\emptyset}{\x}{0}{?}}
\end{array}$$
\end{small}
\caption{Example of stuck derivation}\label{fig:stuck_derivation}
\end{figure}

\begin{figure}
\begin{small}
$$\begin{array}{l}
\infeRule{wd-var}
{\infeRule{wd-cons}
{\infeRule{{wd-nop}}
{\infeRule{wd-corec}
{}
{\WD{\x}{\{\x\mapsto1\}}}
&
\infeRule{wd-var}
{\infeRule{wd-cons}
{\infeRule{wd-corec}
{}
{\WD{\y}{\{\x\mapsto\DA{2},\y\mapsto 1\}}}}
{\WD{\cons{1}{\y}}{\{\x\mapsto1,\y\mapsto0\}}}}
{\WD{\y}{\{\x\mapsto1\}}}}
{\WD{\x\ [+]\ \y}{\{\x\mapsto1\}}}}
{\WD{\cons{0}{(\x\ [+]\ \y{)}}}{\{\x\mapsto0\}}}}
{\WD{\x}{\emptyset}}
\end{array}$$
\\
$$\begin{array}{l}
\infeRule{at-var}
{\infeRule{at-cons-succ}
{\infeRule{at-op}
{\infeRule{at-var}
{\begin{array}{c}{\infeRule{at-cons-0}{}{\At{\mapEnv}{\cons{0}{(\x\ [+]\ \y)}}{0}{0}}}\\\vdots\end{array}}
{\At{\mapEnv}{\x}{\indx-1}{\indx-1}}
&
\infeRule{at-var}
{
\begin{array}{c}{\infeRule{at-cons-0}{}{\At{\mapEnv}{\cons{1}{\y}}{0}{1}}}\\\vdots\end{array}
}
{\At{\mapEnv}{\y}{\indx-1}{{1}}}}
{\At{\mapEnv}{\x\ [+]\ \y}{\indx-1}{{\indx}}}}
{\At{\mapEnv}{\cons{0}{(\x\ [+]\ \y)}}{\indx}{\indx}}}
{\At{\mapEnv}{\x}{\indx}{\indx}}
\end{array}$$
\end{small}
\caption{Derivations for $\mapEnv=\Caps{\x}{\{\x\mapsto \cons{0}{(\x\ [+]\ \y)},\ y\mapsto \cons{1}{\y}\}}$}\label{fig:ex1}
\end{figure}

\begin{figure}
\begin{small}
$$\begin{array}{l}
\infeRule{wd-var}{
\infeRule{wd-cons}{
\infeRule{wd-$\IL$}{
\infeRule{wd-nop}{
\infeRule{wd-corec}{
}
{\WD{\x}{\{\x\mapsto1\}}}\ &
\begin{array}{c}
\vdots \\ \ 
\end{array}
%\infeRule{wd-var}{
%\infeRule{wd-cons}
%{\infeRule{wd-corec}{
%}
%{\WD{\y}{\{\x\mapsto1,\y\mapsto 1\}}}}
%{\WD{\cons{1}{\y}}{\{\x\mapsto1,\y\mapsto0\}}}
%}
%{\WD{\y}{\{\x\mapsto1\}}}
}
{\WD{\x\ [+]\ \y}{\{\x\mapsto1\}}} &
\infeRule{wd-nop}{
\infeRule{wd-corec}{
}
{\WD{\x}{\{\x\mapsto2\}}}\  &
\begin{array}{c}
\vdots \\ \ 
\end{array}
%\infeRule{wd-var}{
%\infeRule{wd-cons}
%{\infeRule{wd-corec}{
%}
%{\WD{\y}{\{\x\mapsto2,\y\mapsto 1\}}}}
%{\WD{\cons{1}{\y}}{\{\x\mapsto2,\y\mapsto0\}}}
%}
%{\WD{\y}{\{\x\mapsto2\}}}
}
{\WD{\x\ [+]\ \y}{\{\x\mapsto2\}}}}
{\WD{(\x\ [+]\ \y)\ {\IL}\ (\x\ [+]\ \y)}{\{\x\mapsto1\}}}}
{\WD{\cons{0}{((\x\ [+]\ \y)\ {\IL}\ (\x\ [+]\ \y))}}{\{\x\mapsto0\}}}}
{\WD{\x}{\emptyset}}
\end{array}$$
\\[4ex]
$$\begin{array}{l}
\infeRule{at-var}{
\infeRule{at-cons-succ}{
\infeRule{at-$\IL$-\{even,odd\}}{
\infeRule{at-nop}{
\infeRule{at-var}{
\vdots}{\At{\mapEnv}{\x}{(\indx-1)/2}{f_i-1}}&
\infeRule{at-var}{
\vdots}{\At{\mapEnv}{\y}{(\indx-1)/2}{1}}
}
{\At{\mapEnv}{\x\ [+]\ \y}{(\indx-1)/2}{f_i}}}
{\At{\mapEnv}{(\x\ [+]\ \y)\ {\IL}\ (\x\ [+]\ \y)}{\indx-1}{f_i}}}
{\At{\mapEnv}{\cons{0}{(\x\ [+]\ \y)\ {\IL}\ (\x\ [+]\ \y)}}{\indx}{f_i}}}
{\At{\mapEnv}{\x}{\indx}{f_i}}
\end{array}$$
\end{small}
\caption{Derivations for ${\mapEnv=\{\x\mapsto \cons{0}{((\x\ [+]\ \y)\ ||\ (\x\ [+]\ \y))}, \y\mapsto 1:\y\}}$ and \mbox{$f_i=\mathtt{bfs\_level()(i)}$, $\indx>1$}}
\label{fig:ex2}
\end{figure}

\pagebreak

\section{Proofs}

\textbf{\refToLemma{basics}}.\
\begin{enumerate}
\item A judgment $\WD{\sv}{\emptyset}$ has no derivation
iff the following condition holds:\\
\begin{tabular}{ll}
\refToRule{wd-stuck}&$\WD{\x}{\map'}\premisestar \WD{\mapEnv(\x)}{\Update{\map}{\x}{0}}\premise\WD{\x}{\map}\premisestar\WD{\sv}{\emptyMap}$\\
&for some $\x\in\dom{\mapEnv}$, 
and $\map',\map$ s.t.\ $\x\not\in\dom{\map}, \map'(\x)\leq 0$.
\end{tabular}
\item If the derivation of $\AT{\sv}{\jndx}$ 
is infinite, then the following condition holds:\\
\begin{tabular}{ll}
\refToRule{at-$\infty$}&$\AT{\x}{\indx+k}\premisestar\AT{\mapEnv(\x)}{\indx}\premise\AT{\x}{\indx}\premisestar\AT{\sv}{\jndx}$\\
&for some $\x\in\dom{\mapEnv}$, and $\indx, k \geq 0$.
\end{tabular}
\end{enumerate}

\begin{proof}%[\refToLemma{basics}]
\begin{enumerate}
\item If $\WD{\sv}{\emptyset}$ has no derivation, since for each $\WD{\sv'}{\map}$ there is an applicable rule, unless in the case $\WD{\x}{\map}$ with $\map(\x)\leq 0$, then condition \refToRule{wd-stuck} holds.
On the other hand, if $\WD{\sv}{\emptyset}$ is derivable, then by \refToLemma{remove-delay} it is also derivable without using \refToRule{wd-delay}. Hence, there can be no path from $\WD{\sv}{\emptyset}$ of judgments on variables in $\dom{\mapEnv}$, and a (first) repeated variable, that is, of the shape below, where $\x\not\in\dom{\map}$: 
\begin{small}
$$\WD{\x}{\map'}\WD{\x_n}{\_}\ldots\WD{\x_1}{\_}\WD{\x}{\map}\WD{\y_m}{\_}\ldots\WD{\y_1}{\_}\ldots\WD{\sv}{\emptyset}$$
\end{small}
\item For each $\AT{\sv}{\indx}$ there is exactly one applicable rule, unless in the case $\AT{\x}{\indx}$ with $\x\not\in\dom{\mapEnv}$. Moreover, since $\mapEnv$ has finite domain, the derivation $\AT{\sv}{\jndx}$ is infinite iff there is a variable, say $\x$, which is repeated infinitely many times, and the indexes cannot be always decreasing.
\end{enumerate}
\end{proof}

\leavevmode

\noindent\textbf{\refToTheorem{wd-eq-gen}}
If $\wdRel{\optmap}{\optmap}{\p}$, then, for all $\sv$, $\WD{\sv}{\map}$ is derivable iff $\owd{\sv}{\optmap}{\p}$ is derivable.
\begin{proof}[$\Rightarrow$] 
By induction on the definition of $\WD{\sv}{\map}$. 
\begin{description}
\item[\refToRule{wd-var}] By hypothesis, $\WD{\mapEnv(\x)}{\Update{\map}{\x}{0}}$ holds with $\x\not\in\dom\map$. By inductive hypothesis, $\owd{\mapEnv(\x)}{\optmap}{\p}$ holds with $\sumFrom(\optmap(x),\p)=0=m(x)$. Thus, we can apply rule \refToRule{owd-var} and have the thesis.
\item[\refToRule{wd-corec}] By hypothesis, $\WD{\x}{\map}$ holds with $\x\in\dom\map$ and $\map(x)>0$. Thus, since $\sumFrom(\optmap(x),\p)>0$ by the definition of $\bowtie$, we have the thesis thanks to rule \refToRule{owd-corec}.
\item[\refToRule{wd-fv}] By hypothesis, $\WD{\x}{\map}$ holds with $\x\not\in\dom\mapEnv$. We immediately have the thesis thanks to rule \refToRule{owd-fv}.
\item[\refToRule{wd-cons}] By hypothesis, $\WD{\sv}{{\incrMap}}$ holds. By inductive hypothesis,\\ $\owd{\sv}{\optmap}{{p\appOp 1}}$ holds. Moreover, for all $\x\in\dom\map$ so that $\incrMap(\x)=i+1$, we have that $\sumFrom(\optmap(x),\p \appOp 1)=i+1$ by definition of $\bowtie$. Thus, we have the thesis thanks to rule \refToRule{owd-cons}.  
\item[\refToRule{wd-tail}] By hypothesis, $\WD{\sv}{{\decrMap}}$ holds. By inductive hypothesis,\\ $\owd{\sv}{\optmap}{{p\appOp (-1)}}$ holds. Moreover, for all $\x\in\dom\map$ so that $\decrMap(\x)=i-1$, we have that $\sumFrom(\optmap(x),\p \appOp (-1))=i-1$ by definition of $\bowtie$. Thus, we have the thesis thanks to rule \refToRule{owd-tail}. 
\item[\refToRule{wd-nop}] By hypothesis, $\WD{\sv_1}{\map}$ and $\WD{\sv_2}{\map}$ hold. By inductive hypothesis, $\owd{\sv_1}{\optmap}{\p}$ and $\owd{\sv_2}{\optmap}{\p}$ hold. Thus, we have the thesis thanks to rule \refToRule{owd-nop}.
{
\item[\refToRule{wd-\DA{$\IL$}}] By hypothesis, $\WD{\sv_1}{\map}$ and $\WD{\sv_2}{\incrMap}$ hold. By inductive hypothesis, $\owd{\sv_1}{\optmap}{\p}$ and $\owd{\sv}{\optmap}{{p\appOp 1}}$ hold. Moreover, for all $\x\in\dom\map$ so that $\incrMap(\x)=i+1$, we have that $\sumFrom(\optmap(x),\p \appOp 1)=i+1$ by definition of $\bowtie$. Thus, we have the thesis thanks to rule \refToRule{owd-$\IL$}.
}
\end{description}
\end{proof}

\begin{proof}[$\Leftarrow$]\\
By induction on the definition of $\owd{\sv}{\optmap}{\p}$. 
\begin{description}
{
\item[\refToRule{owd-var}] By hypothesis, $\owd{\mapEnv(\x)}{\Update{\optmap}{\x}\indx}{\p}$ holds with $\x\not\in\dom\optmap$ and $\l=\len{\p}$. By inductive hypothesis, $\WD{\mapEnv(\x)}{\Update{\map}{\x}{0}}$ holds with $\sumFrom(\optmap(x),p)=0=m(x)$. Thus, we can apply rule \refToRule{wd-var} and have the thesis.
\item[\refToRule{owd-corec}] By hypothesis, $\owd{\x}{\optmap}{\p}$ holds with $\x\in\dom\optmap$ and\\ $
\sumFrom(\map(\x),\p)>0$. Thus, since $\map(x)>0$ by the definition of $\bowtie$, we have the thesis thanks to rule \refToRule{wd-corec}.
\item[\refToRule{owd-fv}] By hypothesis, $\owd{\x}{\optmap}{\p}$ holds with $\x\not\in\dom\mapEnv$. We immediately have the thesis thanks to rule \refToRule{wd-fv}.
\item[\refToRule{owd-cons}] By hypothesis $\owd{\sv}{\optmap}{{p\appOp 1}}$ holds. By inductive hypothesis \\ $\WD{\sv}{{\incrMap}}$ holds. Moreover, for all $\x\in\dom\map$ so that $\sumFrom(\optmap(x),\p \appOp 1)=i+1$, we have that $\incrMap(\x)=i+1$ by definition of $\bowtie$. Thus, we have the thesis thanks to rule \refToRule{wd-cons}.
\item[\refToRule{owd-tail}] By hypothesis, $\owd{\sv}{\optmap}{{p\appOp (-1)}}$ holds. By inductive hypothesis, $\WD{\sv}{{\decrMap}}$ holds. Moreover, for all $\x\in\dom\map$ so that $\sumFrom(\optmap(x),\p \appOp (-1))=i-1$, we have that $\decrMap(\x)=i-1$ by definition of $\bowtie$. Thus, we have the thesis thanks to rule \refToRule{wd-tail}.
\item[\refToRule{owd-nop}] By hypothesis, $\owd{\sv_1}{\optmap}{\p}$ and $\owd{\sv_2}{\optmap}{\p}$ hold. By inductive hypothesis, $\WD{\sv_1}{\map}$ and $\WD{\sv_2}{\map}$ holds. Thus, we have the thesis thanks to rule \refToRule{wd-nop}.
\item[\refToRule{owd-$\IL$}] By hypothesis, $\owd{\sv_1}{\optmap}{\p}$ and $\owd{\sv}{\optmap}{{p\appOp 1}}$ hold. By inductive hypothesis, $\WD{\sv_1}{\map}$ and $\WD{\sv_2}{\incrMap}$ hold. Moreover, for all $\x\in\dom\map$ so that $\sumFrom(\optmap(x),\p \appOp 1)=i+1$ , we have that $\incrMap(\x)=i+1$ by definition of $\bowtie$. Thus, we have the thesis thanks to rule \refToRule{wd-$\IL$}.
}
\end{description}
\end{proof}

\end{document}